\documentclass[10pt,letterpaper]{article}
\usepackage[margin=1in]{geometry}

\usepackage{setspace}
\usepackage{latexsym}

\usepackage{booktabs}   %% For formal tables:
\usepackage{subcaption} %% For complex figures with subfigures/subcaptions
\usepackage{graphicx}
\usepackage{float}
\usepackage{amsmath}
\usepackage{amsthm}
\usepackage{amssymb}
\usepackage{afterpage}
\usepackage{epsfig}
\usepackage[shortlabels]{enumitem}
\setlist[enumerate]{leftmargin=*}\usepackage{longtable}
\usepackage{csquotes}
\usepackage{textcomp}
\PassOptionsToPackage{hyphens}{url}\usepackage{hyperref}
\usepackage{color}
\usepackage{setspace}
\usepackage{titlesec}
\definecolor{linkColor}{rgb}{0.,0.11,0.22}
\definecolor{YaleBlue}{rgb}{0.0,0.22,0.444}
\usepackage{bbm}
\usepackage{xcolor}
\usepackage{ragged2e}
\usepackage{subcaption}
\usepackage{algorithm}
\usepackage[noend]{algpseudocode}
\usepackage{mathtools}
\usepackage{pifont}
\newcommand{\cmark}{\ding{52}}%
\newcommand{\xmark}{\ding{56}}%
\usepackage{tabularx, booktabs}
\newcolumntype{Y}{>{\centering\arraybackslash}X}
\usepackage{hhline}
\usepackage{listings}
\usepackage{titling}

\newtheorem{thm}{Theorem}[section]
\newtheorem{defi}[thm]{Definition}

\definecolor{codegreen}{rgb}{0,0.6,0}
\definecolor{codegray}{rgb}{0.5,0.5,0.5}
\definecolor{codepurple}{rgb}{0.58,0,0.82}
\definecolor{backcolour}{rgb}{0.97,0.97,0.97}
\definecolor{darkbackcolour}{rgb}{0.875,0.875,0.875}
\newcommand{\code}[1]{\colorbox{darkbackcolour}{\texttt{#1}}}

\usepackage{ebgaramond}
\newcommand*{\garamond}{\fontfamily{EBGaramond-OsF}\selectfont}

\newlength\longest
\newlength\nlongest

\DeclareMathAlphabet{\mathdutchcal}{U}{dutchcal}{mb}{n}

\usepackage{pgfplots}
\pgfplotsset{compat=1.15}
\usepackage{mathrsfs}
\usepackage{xspace}
\usetikzlibrary{arrows, patterns}
\definecolor{ccwwqq}{rgb}{0.8,0.4,0}
\definecolor{qqwwzz}{rgb}{0,0.4,0.6}
\definecolor{ffffff}{rgb}{1,1,1}
\definecolor{wqwqwq}{rgb}{0.3764705882352941,0.3764705882352941,0.3764705882352941}
\definecolor{yqyqyq}{rgb}{0.5019607843137255,0.5019607843137255,0.5019607843137255}
\definecolor{ttzzqq}{rgb}{0.2,0.6,0}
\definecolor{xdxdff}{rgb}{0.49019607843137253,0.49019607843137253,1}
\definecolor{ududff}{rgb}{0.30196078431372547,0.30196078431372547,1}
\definecolor{zzttqq}{rgb}{0.6,0.2,0}
\definecolor{qqwuqq}{rgb}{0,0.39215686274509803,0}
\definecolor{qqttcc}{rgb}{0,0.2,0.8}
\definecolor{ccxzax}{rgb}{0.8,0.4745098039215686,0.6549019607843137}
\definecolor{ewzfqq}{rgb}{0.9019607843137255,0.6235294117647059,0}
\definecolor{vwbuez}{rgb}{0.33725490196078434,0.7058823529411765,0.9137254901960784}
\definecolor{codegreen}{rgb}{0,0.6,0}
\definecolor{codegray}{rgb}{0.5,0.5,0.5}
\definecolor{codepurple}{rgb}{0.58,0,0.82}
\definecolor{backcolour}{rgb}{1,1,1}
\definecolor{OliveGreen}{rgb}{0.122,0.53,0}
\usepackage{listings}

\lstdefinestyle{customC}{
    language=C,
    backgroundcolor=\color{backcolour},   
    commentstyle=\color{codegreen},
    keywordstyle=\color{teal},
    stringstyle=\color{codepurple},
    basicstyle=\ttfamily\scriptsize,
    breakatwhitespace=false,
    xleftmargin=12pt,
    numbers=none,
    numbers=left,
    numbersep=4pt,
    upquote=true
    frame=lines,
    breaklines=true,                 
    captionpos=b,                    
    keepspaces=true,
    showspaces=false,                
    showstringspaces=false,
    showtabs=false,
    tabsize=2,
    morekeywords={bool},
    framexleftmargin=0mm,
}

\lstdefinestyle{customPy}{
    language=Python,
    backgroundcolor=\color{backcolour},   
    commentstyle=\color{codegreen},
    keywordstyle=\color{teal},
    stringstyle=\color{codepurple},
    basicstyle=\ttfamily\scriptsize,
    breakatwhitespace=false,
    xleftmargin=12pt,
    numbers=none,
    numbers=left,
    numbersep=4pt,
    upquote=true,
    frame=lines,
    breaklines=true,                 
    captionpos=b,                    
    keepspaces=true,
    showspaces=false,                
    showstringspaces=false,
    showtabs=false,
    tabsize=2,
    framexleftmargin=0mm,
    morekeywords={True, False},
    deletekeywords={from},
    morestring=[s][\color{OliveGreen}]{@soid}{register}
}

\lstdefinestyle{tinyPy}{
    language=Python,
    backgroundcolor=\color{backcolour},   
    commentstyle=\color{codegreen},
    keywordstyle=\color{teal},
    stringstyle=\color{codepurple},
    basicstyle=\ttfamily\tiny,
    breakatwhitespace=false,
    xleftmargin=12pt,
    numbers=none,
    numbers=left,
    numbersep=4pt,
    upquote=true,
    frame=lines,
    breaklines=true,                 
    captionpos=b,                    
    keepspaces=true,
    showspaces=false,                
    showstringspaces=false,
    showtabs=false,
    tabsize=2,
    framexleftmargin=0mm,
    morekeywords={True, False},
    deletekeywords={from},
    morestring=[s][\color{OliveGreen}]{@soid}{register}
}

\newcommand{\AC}[1]{{#1}}
\newcommand{\RA}[1]{{#1}}
\newcommand{\RB}[1]{{#1}}
\newcommand{\RC}[1]{{#1}}

\newenvironment{caquote}%
  {\list{}{\leftmargin=0.65in\rightmargin=0.65in}\item[]}%
  {\endlist}

\newcommand{\counterfactual}{\ensuremath{%          
    \mathrel{\Box\kern-1.5pt\raise0.8pt\hbox{\vspace{10pt}$\mathord{\rightarrow}$}}}}
\newcommand{\fprop}{\ensuremath{\to_{\mathcal{A}, t, \ell}}}
\newcommand{\cprop}{\ensuremath{\counterfactual_{\mathcal{A}, t^*, \ell}}}

\newcommand{\rsigma}{\left.\sigma\right|}

\newcommand{\soid}{$\textsf{soid}$}
\newcommand{\longname}{Counterfactual-Guided Logic Exploration and Abstraction Refinement}
\newcommand{\acronym}{CLEAR}

\begin{document}

\newlength{\blob}
\settowidth{\blob}{Graz University of Technology}
\title{\vspace{-15mm}{\huge`Put the Car on the Stand': SMT-based Oracles for Investigating Decisions}}
\author{
\begin{tabular}{c c c}
Samuel Judson & Matthew Elacqua & Filip Cano    \\
\makebox[\blob][c]{Yale University} & \makebox[\blob][c]{Yale University} & Graz University of Technology \\
{\footnotesize\url{samuel.judson@yale.edu}} & 
{\footnotesize\url{matt.elacqua@yale.edu}} &
{\footnotesize\url{filip.cano@iaik.tugraz.at}} 
\end{tabular}
\\ \\
\begin{tabular}{c c}
Timos Antonopoulos & Bettina K\"onighofer \\
Yale University & Graz University of Technology  \\
{\footnotesize\url{timos.antonopoulos@yale.edu}} &
{\footnotesize\url{bettina.koenighofer@iaik.tugraz.at}} \\
& \\
Scott J. Shapiro & Ruzica Piskac \\
Yale Law School \& Yale University & Yale University \\
{\footnotesize\url{scott.shapiro@yale.edu}} &
{\footnotesize\url{ruzica.piskac@yale.edu}}
\end{tabular}
}
\date{}
\maketitle
\begin{abstract}
\noindent Principled accountability in the aftermath of harms is essential to the trustworthy design and governance of algorithmic decision making. Legal theory offers a paramount method for assessing culpability: putting the agent `on the stand' to subject their actions and intentions to cross-examination. We show that under minimal assumptions automated reasoning can rigorously interrogate algorithmic behaviors as in the adversarial process of legal fact finding. We model accountability processes, such as trials or review boards, as \longname{} (\acronym{}) loops.
We use the formal methods of symbolic execution and satisfiability modulo theories (SMT) solving to discharge queries about agent behavior in factual and counterfactual scenarios, as adaptively formulated by a human investigator. In order to do so, for a decision algorithm $\mathcal{A}$ we use symbolic execution to represent its logic as a statement $\Pi$ in the decidable theory \texttt{QF\_FPBV}. 
We implement our framework 
and demonstrate its utility on an illustrative car crash scenario.
\end{abstract}

\section{Introduction}\label{sec:intro}

Our lives are increasingly impacted by the automated decision making of AI. We share roads with autonomous vehicles, as healthcare providers use algorithms to diagnose diseases and prepare treatment plans, employers to automate hiring screens, and even judges to analyze flight and recidivism risks. Though the creators of AI often intend it to improve human welfare, it is a harsh reality that algorithms often fail. Automated decision makers (ADMs) are now deployed into roles of immense social responsibility even as their nature means they are not now, and likely will never be, trustworthy enough to do no harm. When autonomous vehicles drive on open roads they cause fatal accidents~\cite{smiley2022m}. Classification and scoring algorithms perpetuate race- and gender-based biases in hiring and recidivism evaluations, both characteristics legally protected from discrimination in many countries~\cite{angwin2016machine,dastin2018amazon, kroll2017accountable}. Both the rule of law and a more ordinary sense of justice demand that society hold accountable those responsible and answerable for harms. In this work, we investigate how formal methods can aid society and the law in providing accountability
and trust in a clear, rigorous, and efficient manner through SMT-based automated reasoning.

Ideally, computer scientists would verify decision making algorithms to confirm their correctness before deployment. Using the techniques of program verification discrimination and other social harms would be automatically detected and eliminated by engineers before ADMs appear in the field. Unfortunately, the formal verification of these algorithms is undecidable in the general case, and even when theoretically possible will often require computational power that can make the task uneconomical, if not practically infeasible. Additionally, writing specifications for algorithmic decision making often treads onto contentious questions of law and policy with no universally agreed upon, let alone formalizable, answers~\cite{kroll2017accountable}. Nevertheless, assessing responsibility for harms remains vital to the safe use of ADMs.

\RC{To better understand the concept of accountability, consider a case in which one autonomous car hits another. We can ask: Which car is responsible for the accident? Which made the error, and to what end? When human drivers crash, lawyers investigate the drivers’ reasoning and actions. Did the drivers intend to hit the other car? Did the drivers know that an accident would occur? To infer drivers’ intentions, lawyers engage in direct and indirect examination to uncover the decision logic that lead to the accident.} 
Standard notions of due process view this opportunity to mount a defense as essential to justice. A person must be permitted to explain and justify themselves through arguing the facts of the case~\cite{sep-criminal-law}.
For example, if a human driver can convince a jury the crash was an unforeseeable accident, then they may be subject to lesser penalties than had they acted intentionally or negligently. A self-driving car cannot do likewise. Like a person, an ADM makes unsupervised decisions in complex environments which can lead to harm. But very much unlike a person, an algorithm cannot simply walk into a courtroom and swear to tell the whole truth. 

ADMs leave us with the traditional need for explanation but -- as a program and not a mind -- without the traditional means for acquiring one. Still, a program can be translated into logic, and logic can be rigorously reasoned about. While an algorithm may not be able to defend itself under interrogation, we show the right formal method can absolutely advocate on its behalf. And at least in one sense, a formal method makes for a better witness than a human -- while the human can lie, the car cannot. The provable rigor of our approach guarantees that whatever answers we get from the decision algorithm are both accurate and comprehensive.

\subsection{Contribution}

We developed a method \acronym{} and tool, \soid{}, for applying automated reasoning to `put the algorithm on the stand' in cases where its correct behavior cannot be reduced to a practically verifiable formal specification. Using \soid{}{}, an \emph{investigator} can pose tailored \emph{factual} and \emph{counterfactual} queries to better understand the functional intention of the
decision algorithm, in order to distinguish accidents from honest design failures from malicious design consequences. Our method also generates counterexamples and counterfactuals
to challenge flawed conclusions about agent behavior – just as would a human assisting
in their own defense. We assume only access to a program $A$ implementing a decision making algorithm $\mathcal{A}$ as granted by its \emph{controller}. Our method supports three types of queries: \emph{factuals} (`did the agent do...'), \emph{might counterfactuals} (`might the agent have possibly done...'), and \emph{would counterfactuals} (`would the agent have necessarily done...')~\cite{lewis2013counterfactuals}. We formally define each in~Section~\ref{sec:prelims}. The distinction between `would' and `might' counterfactuals is foundational to their treatment in philosophy~\cite{lewis2013counterfactuals}. Logically, we implement `might' counterfactuals using the $\exists$ operator, and `would' counterfactuals using $\forall$.

We have also implemented an example of a domain-specific graphical user interface (GUI) that allows operation of \soid{} without requiring technical expertise, but here describe how it works directly. The investigator starts from the logs of $A$, the factual information within capturing the state of the world as the agent perceived it. These logs can be easily translated into a first-order formalism as a sequence of equalities. For example, in the left panel of the car crash diagrammed in Figure~\ref{fig:crash} the information in the logs of $A$ (blue at bottom) about the other car (red at left) might be encodable as a system of equalities
\begin{equation}\label{eqn:fact}
  \varphi \equiv \texttt{agent1\_pos\_x} = 1.376 \land \cdots \land \texttt{agent1\_signal} = \text{RIGHT} \land \cdots.  
\end{equation}
In the GUI all such statements are translated into formal logic automatically. Using the automated reasoning method of symbolic execution, we can answer queries about what the car did under these constraints by checking formula entailment. Posing counterfactual queries requires
manipulating the state of the world since such queries ask how would the agent react
if the situation were different. Counterfactuals can be encoded by substituting constraint values
\begin{equation}\label{eqn:cfact}
\varphi' \equiv \varphi[(\texttt{agent1\_pos\_x} = 1.376) \mapsto (\texttt{agent1\_pos\_x} = 1.250)].
\end{equation}
This formalism also allows us to reason about whole families of counterfactuals, which can be defined by relaxing the constraints and negating the original factual state as a valid model
\begin{equation}\label{eqn:punc}
\varphi'' \equiv \varphi[(\texttt{agent1\_pos\_x} = 1.376) \mapsto (1.0 \leq \texttt{agent1\_pos\_x} \leq 1.5) ] \land \lnot \varphi.
\end{equation}
In this way, hypothetical-but-similar scenarios of interest to the investigator (`what if that car was outside instead of inside the intersection?', `what if the car was signaling a turn instead of straight?') can be rigorously formalized to enable automated analysis of the agent's behavior. 

\captionsetup[figure]{format=hang}
\begin{figure}[!t]
     \centering
     \begin{subfigure}[t]{0.396\textwidth}
        \begin{flushright}
            \includegraphics[width=\textwidth, trim={1.63cm 2.5cm 1.63cm 2.25cm},clip]
 {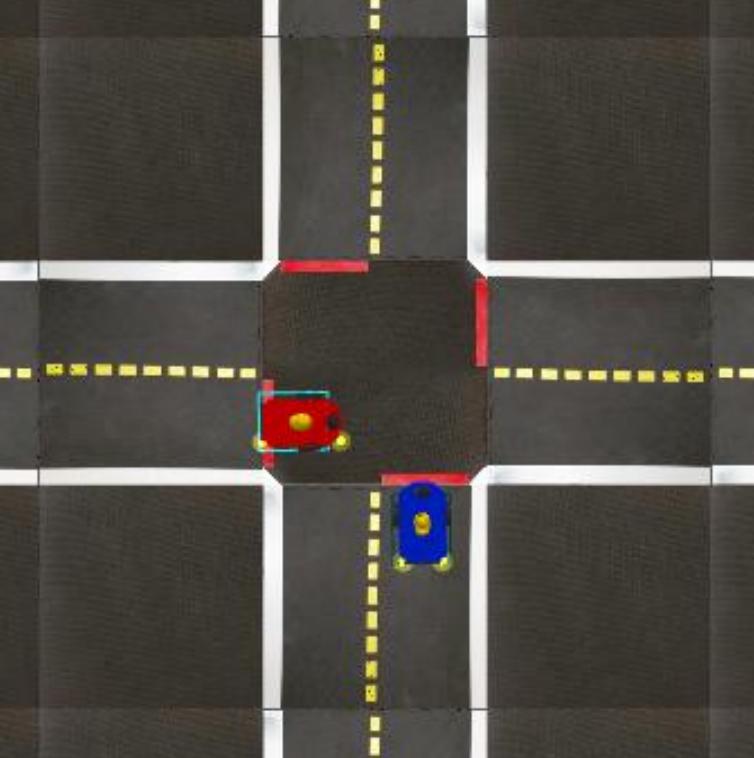}
        \end{flushright}
     \end{subfigure}
     \hspace{5mm}
     \begin{subfigure}[t]{0.4\textwidth}
          \begin{flushleft}
            \includegraphics[scale=0.285, trim={1.85cm 2.3cm 1.95cm 2.45cm},clip]{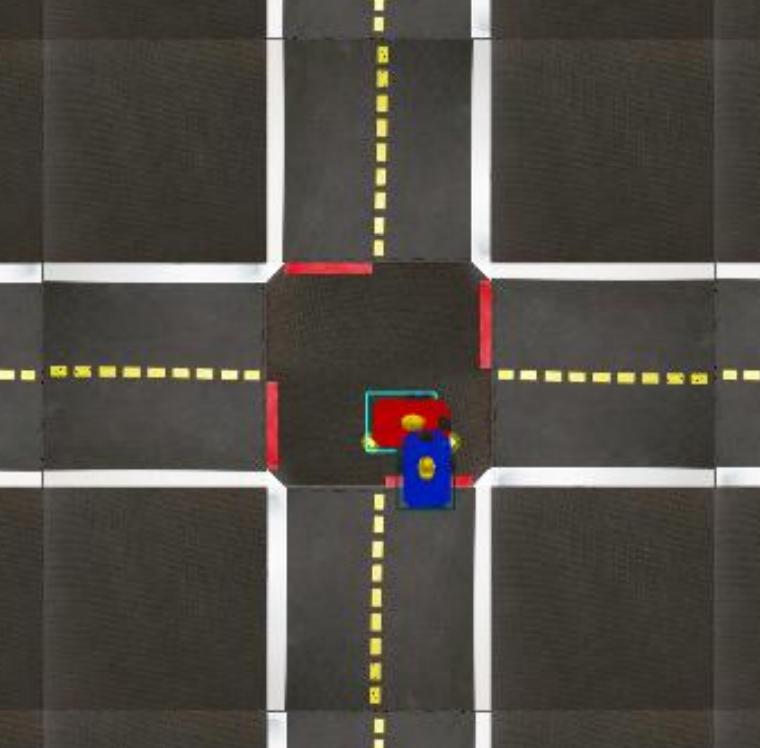}  
          \end{flushleft}
     \end{subfigure}
     \caption{A broadside car crash rendered in the \soid{} GUI. \label{fig:crash}}
\end{figure}

We model the adaptive, semi-automated process \soid{} implements as a \emph{counterfactual-guided logic exploration and abstraction refinement} (\acronym{}) loop. We use SMT-based automated reasoning for the oracle that answers the factual and counterfactual questions the investigator asks. We rely on the Z3 SMT solver~\cite{moura2008z3} and -- for our benchmarks -- KLEE-Float for symbolic execution~\cite{cadar2008klee,liew2017floating}, though through a modular API \soid{} supports any symbolic execution engine able to generate output in the SMT-LIB format. We demonstrate \soid{} on a car crash scenario caused by a reinforcement-learned agent within a simulated traffic environment, and show that \soid{} abstracts away sufficient detail of the code of the simulated autonomous vehicle to be used through an intuitive GUI by investigators without technical expertise.

Our \acronym{} loop is inspired by the CEGAR loop~\cite{clarke2000counterexample}. The main difference is that, while the CEGAR loop produces counterexamples automatically, we rely on the human experts to formulate the queries. Each query and its answer helps the expert investigator to build a body of knowledge, the $\textsc{Facts}$, about the decision
logic used in the autonomous agent. The investigator also decides when to terminate the \acronym{} loop and when they have enough information for the final judgment of an agent's culpability.

\paragraph{Outline.} In~\S\ref{sec:motivating}  we discuss the motivation for \soid{} and give a few examples that also underlie our benchmarks. In~\S\ref{sec:prelims} we overview SMT-based automated reasoning and counterfactual theories of causation and responsibility. In~\S\ref{sec:clear} we present our \acronym{} procedure, before in~\S\ref{sec:rsandqs} detailing how we logically encode agent decision making and counterfactual scenarios. We then describe the use of the \soid{} tool through quantitative and qualitative evaluation in~\S\ref{sec:soid}, before concluding in~\S\ref{sec:conclusion}. 

\subsection{Related Work}\label{sec:rw}

Formal methods for accountability are a burgeoning research topic, both in general~\cite{baier2021verification,feigenbaum2020accountability,datta2015program,kusters2010accountability,halpern2020causes1,halpern2020causes2,chockler2004responsibility} and focused on specific domains including automated and economic decision making~\cite{kroll2017accountable,baier2021game,baier2021responsibility,su2015interpretable,ghosh2019imli} and security~\cite{feigenbaum2011towards,kunnemann2019automated}. In particular, a recent work uses such methods to investigate similar questions to ours under stronger modeling assumptions~\cite{cordoba2023analyzing}. Trustworthy algorithmic decision making is now a major focus of classical formal methods research as well~\cite{garcia2015comprehensive,katz2019marabou,singh2019abstract,gehr2018ai2,alshiekh2018safe,dreossi2019verifai,christakis2021automated}. Intention and its relationship to responsibility is a central focus of law and the philosophy of action, with a cross-discipline history dating back millennia~\cite{moore2019causation,beebee2019counterfactual,sep-counterfactuals,lewis2013counterfactuals,wachter2017counterfactual}, including an extensive modern focus on (often symbolic) automated decision making in particular~\cite{bratman1987intention,cohen1990intention,rao1991modeling}. Counterfactual reasoning is accordingly a significant topic in Explainable AI (XAI)~\cite{arrieta2020explainable,adadi2018peeking,guidotti2018survey,padovan2023black}, using both logical~\cite{halpern2020causes1,halpern2020causes2, chockler2004responsibility} and statistical methods~\cite{wachter2017counterfactual}. SMT solving~\cite{moura2008z3,barrett2010smt} and symbolic execution~\cite{SurveySymExec-CSUR18,cadar2008klee}, are both foundational topics in automated reasoning.

\AC{Technically, our approach differs from the already significant body of work in counterfactual analysis of algorithmic decision making in two significant ways: in our analysis of executable code `as it runs' -- rather than just of a partial component (such as an ML model in isolation) or of some higher-level mathematical abstraction of the system -- and in our reliance on formal verification. By analyzing the code itself, rather than an idealized abstraction or particular model, we can capture behaviors of the entire software system: preprocessing, decision making, postprocessing, and any bugs and faults therein. This makes our analysis more complete. In~Section~\ref{app:dt}, for example, we consider a hypothetical case study where an instance of API misusage -- rather than any mistakes of logic within the code itself -- undermines a machine learning decision. The decision and its consequences are not analyzable by considering only the (correct on its own terms) decision model alone.}

\AC{Meanwhile, verification technologies allow us to analyze \textit{all possible executions} obeying \textit{highly expressive pre- and post-conditions}. SMT-based methods in particular provide the full expressiveness of first-order logic. As such, our approach can encode entire families of counterfactuals in order to provide a broad and thorough picture of agent decision making, so as to better interpret responsibility. Prevailing, often statistical methods, commonly focus more on gathering explanations prioritized on informal measures, such as minimality or diversity criteria, in order to demonstrate causality, see \textit{e.g.},~\cite{wachter2017counterfactual, mothilal2020explaining}. Informally, in computer science terms, this distinction is analogous to that in automated reasoning between methods for verification -- which emphasize overall safety and the correctness of the set of all program traces provably meeting some logically expressed property -- and methods like testing or fuzzing that focus on finding or excluding representative executions believed to exemplify that property~\cite{abebe2022adversarial}. Of course (SMT-based) verification does have costs -- it carries substantially more computational complexity than testing approaches, which can increase compute costs and limit scalability. Accordingly, we implement and benchmark the empirical efficacy of our method in a laboratory environment in~Section~\ref{sec:soid}. In human terms, our approach is analagous to enabling asking broad, positive questions about agent behavior under a coherent family of scenarios, rather than asking questions aimed primarily at generating or falsifying a particular claimed explanation for a (factual or counterfactual) decision. The work of~\cite{cordoba2023analyzing} and VerifAI~\cite{dreossi2019verifai} are the related approaches of which we are aware most similar to our own in goals and method, although the former requires stronger modeling assumptions while the latter sacrifices some formal guarantees for scalability.}

Our work also joins recent efforts to apply formal methods to legal reasoning and processes~\cite{merigoux:hal-03159939,delaet:hal-03447072,satoh2010proleg,khoja22}. The recent workshop series on Programming Languages and the Law (ProLaLa) has featured various domain-specific languages for formal reasoning about laws and regulations, such as formalizing norms used in privacy policy regulations, 
legal rental contracts, 
and property law, 
just to name a few.
The need for a formal definition of 
accountability and establishing a connection between software systems and legal systems has also been recognized in the Designing Accountable Software Systems (DASS) program, 
started by the National Science Foundation (NSF), the main funding body for foundational 
research in the USA. In his influential CACM column, Vardi~\cite{DBLP:journals/cacm/Vardi22e} has advocated for the importance of accountability in the computing marketplace.

\section{Motivation}\label{sec:motivating}

To illustrate the design and purpose of \soid{}{}, we continue with the crash from Figure~\ref{fig:crash}. In the left panel the autonomous vehicle $\mathcal{A}$ (blue at bottom) perceived that the other car (red at left) had its right turn signal on. Call this time $t^*$. When $\mathcal{A}$ entered the intersection -- believing the action to be safe, even though it lacked the right-of-way -- the other car proceeded straight, leading to a collision. Because it did not possess the right-of-way, $\mathcal{A}$ is culpable for the crash. This scenario forms the basis for our benchmarks in~Section~\ref{sec:soid}, where the specific question we investigate with $\soid{}$ is `with what purpose did $\mathcal{A}$ move, and so to what degree is it culpable for the crash?'

Note that the actions of $\mathcal{A}$ are consistent with with three significantly different interpretations: 
\begin{enumerate}[i]
    \item a \textbf{reasonable} (or \textbf{standard}) $\mathcal{A}$ drove carefully, but proceeded straight as is common human driving behavior given the (perception of an) indicated right turn;
    \item an \textbf{impatient} $\mathcal{A}$ drove with reckless indifference to the risk of a crash; and
    \item a \textbf{pathological} $\mathcal{A}$ drove to opportunistically cause crashes with other cars, without unjustifiable violations of traffic laws such as weaving into an oncoming lane.\footnote{Such an agent might be attempting a staged crash for insurance fraud.} If no opportunity presented itself, the car would move on.
\end{enumerate}
Even for the same act, these different interpretations will likely lead to drastically different liabilities for the controller under criminal or civil law. Interestingly, the natural language explanations for the i) \emph{reasonable} and iii) \emph{pathological} cars are identical: `moving straight would likely not cause a crash, so proceeding would bring me closer to my goal'. Nonetheless, counterfactual queries (notated $\counterfactual$) can help distinguish between these candidate interpretations. 
\begin{enumerate}
 \item[at] $t^*$\\
     \vspace{-7mm}
    \begin{quote}
        $\counterfactual$ Could a different turn signal have led $\mathcal{A}$ to remain stationary? 
    \end{quote}
    \vspace{-2mm}
    \begin{quote}
        $\counterfactual$ If $\mathcal{A}$ had arrived before the other car, and that other car was not signaling a turn, would $\mathcal{A}$ have waited? (\emph{e.g.}, to `bait' the other car into passing in front of it?)   
    \end{quote}
\end{enumerate}
Interpretation i) is consistent with (\emph{yes}, \emph{no}), ii) with (\emph{no}, \emph{no}), and iii) with (\emph{no}, \emph{yes}). Note the adaptive structure of our questions, where the second query can be skipped based on the answer to the first. The goal of $\soid{}$ is to enable efficient and adaptive investigation of such queries, in order to distinguish the computational reasoning underlying agent decisions and support principled assessment of responsibility. Although autonomous vehicles provide an insightful example, \acronym{} is not limited to cyberphysical systems. In~Section~\ref{app:dt} we use \soid{} to analyze a buggy application of a decision tree leading to a health risk misclassification. We also give another, more nuanced motivating self-driving car crash example in Appendix~\ref{app:motivate}, drawn directly from the self-driving car and smart road network industries.

\subsection{Legal Accountability for ADMs}\label{sec:la}

Before presenting the technical details of \soid{}, we also overview how the philosophy and practice of legal accountability might apply to ADMs, and so motivate the analysis \soid{} is designed to enable. We work in broad strokes as the legal liability scheme for ADMs is still being developed, and so we do not want to limit our consideration to a specific body of law. This in turn limits our ability to draw specific conclusions, as disparate bodies of law often place vastly different importance on the presence of intentionality, negligence, and other artifacts of decision making.

A core principle of legal accountability is that the `why' of a wrongful act is almost always relevant to evaluating how (severely) liable the actor is. In the words of the influential United States Supreme Court justice Oliver Wendell Holmes, `even a dog distinguishes between being stumbled over and being kicked.' As every kick has its own reasons bodies of law often distinguish further -- such as whether the `why' is an active intent to cause harm. Though holding algorithmic agents accountable raises the many technical challenges that motivate \soid{}, once we understand the `why' of an algorithmic decision we can still apply the same framework of our ethical and legal practices we always use for accountability~\cite{kroll2017accountable,hallevy2013robots}. The algorithmic nature of a harmful decision does not invalidate the need for accountability: the locus of Holmes' adage lies in the harm to the victim, being justifiably more aggrieved to be injured on purpose or due to a negligent disregard of the risk of a kick than by an accidental contact in the course of reasonable behavior. In practice, even though criminal law and civil law each place different emphasis on the presence of attributes like intention and negligence in a decision, intention in particular almost always matters to -- and often intensifies -- an agent's liability. Any time the law penalizes an unintentional offense it will almost always punish an intentional violation as well, and should intention be present, the law will usually apply the greatest possible penalties authorized for the harm. Given the importance of recognizing intention, \soid{} is designed to support rigorous and thorough \textit{findings of fact} about algorithmic decisions from which a principled assessment of their `why' can be drawn. 

Taking a step back, it is deeply contentious whether ADMs now and in the future can, could, or should possess agency, legal personhood, or sovereignty, and whether they can ever be morally and legally responsible~\cite{sep-ethics-ai}. Even the basic nature of computational decision making is a significant point of debate in artificial intelligence and philosophy, with a long and contentious history~\cite{bratman1987intention,sep-chinese-room}. For the moment, ADMs are not general intelligences. They will likely not for the foreseeable future possess cognition, agency, values, or theory of mind, nor will they formulate their own goals and desires, or be more than `fancy toasters' that proxy the decision making agency and responsibility of some answerable controller. An algorithm is no more than a computable function implemented by symbolic manipulation, statistically-inferred pattern matching, or a combination thereof. Nonetheless, even working off the most stringent rejection of modern ADMs forming explicit knowledge or intentional states, following Holmes we can see there is still value in grading the severity of a harmful decision. It is deeply ingrained in our governing frameworks for legal and moral accountability that when acting with the purpose of harm an agent (or its controller) has committed a greater transgression than in the case where the harm was unintended.

In this work we sidestep whether and how computers can possess intentionality by viewing intention through a functionalist lens. Even for conscious reasoning, it is impossible to replay a human being's actual thought process during a trial. So in practice, legal definitions refer instead to an \emph{ex post} rationalization of the agent's decisions made by the accountability process through the finding of fact. A person is assessed to have, \emph{e.g.}, purposely caused harm if the facts show they acted in a way that is consistent with purposeful behavior. We can approach computational reasoning in much the same way, with an investigator making an \emph{ex post} descriptive rationalization capturing their understanding of an ADM's decision making. This understanding then justifies a principled assessment of the controller's responsibility. For example, a controller can be assessed to have released into the world an ADM that the facts show acted in a way that is consistent with a purposeful attempt to cause harm. The design and algorithmic processes of the agent are otherwise irrelevant. How the ADM actually decision makes -- whether through statistical inference or explicit goal-oriented decision logic or otherwise -- is relevant only with respect to our ability to interrogate its decision making. This approach is consistent with \soid{}, which is capable of analyzing arbitrary programs.

\AC{An investigator using \soid{} to label an ADM as `reasonable' or `reckless' or `pathological' or similar is, however, only the start. How such an assessment should then be interpreted and used by an accountability process is, ultimately, a policy question. The unsettled nature of the laws, policies, and norms that govern ADMs, both for now and into the future, means there are many open questions about the relevance of the intent of an ADM and its relationship to the intent of the controller. But we can consider the ramifications in broad strokes. For individuals harmed by ADMs (whether as consumers, other end-users, or just unlucky `bystanders'), the situation seems little different than for human misconduct: the finding of intent amplifies the harm, and the victim can reasonably expect the accountability process to penalize the transgressor appropriately. More specific questions are harder. Should apparent intent in both the controller and ADM be assessed more harshly than in one or the other alone? Or would apparent intent in the controller render the actions of the ADM relevant only in how successfully the intent of the controller was carried out? How should an emergent `algorithmic intent' traceable to software faults interact with any documented, contrary evidence of the intent of the controller? These questions lay beyond the scope of this work, but they are each dependent on our capacity to first recognize and distinguish the functional intent of the ADM, motivating our research goals.} 

\AC{For the controllers of ADMs (whether as programmers, vendors, owners, or sovereign states), it is a natural starting point to view them as responsible for the actions of their computational agents, just as they would (most often) be responsible for human agents acting on their behalf. With reward comes responsibility. If a controller profits from deploying an ADM, so must they bear the costs of its harms. Legal concepts governing humans acting on behalf or through each other or organizations are well-founded throughout, \textit{e.g.}, agency and criminal law~\cite{hallevy2013robots,cornellLII}. These mechanisms may be either directly applicable or can form the basis for analogous systems governing algorithmic accountability. For example, just as a business is expected to adequately prepare (\text{i.e.}, train) a human agent to operate on their behalf without causing harm, a controller can be expected to adequately prepare (\textit{i.e.}, design or train) a computational agent. What standard the controller sets internally \textit{ex ante} before deploying the ADM is primarily relevant insofar as it provides confidence to the controller the ADM will not be found \textit{ex post} to have operated in a way consistent with an intent to harm -- and so carry with it a corresponding increase in liability.} 

\AC{Grounding our approach in the functionalist perspective also helps us manage difficult questions about the validity of anthropomorphizing algorithmic systems through the use of language like `intent', `beliefs', or `reasonableness', as we ourselves have done throughout~Section~\ref{sec:motivating}. It is not immediately clear such language is intrinsically confusing or harmful: the use of such labels in characterizing automated decision making is decades-old, to the extent that consideration of whether and how machines can form intentional states has informed how prevailing approaches in the philosophy of action now capture whether and how humans form them~\cite{bratman1987intention, bratman1988plans}. Moreover, as accountability processes begin to wrestle with algorithmic decision making some anthropomorphization is perhaps unavoidable, due to the often heavily analogical nature of legal reasoning~\cite{levi1947introduction}. We ourselves invoked the analogy of Holmes to frame our discussion. The validity of some such analogies are in some cases already contentious. For example, whether the `creativity' required to earn authorship under copyright law must necessarily be human is under active consideration in litigation and scholarship concerning generative models~\cite{dornis2020artificial, thaler}. On the other hand, the negative consequences of anthropomorphizing ADMs has been itself widely recognized in scholarship and science fiction dating back decades: it can cause us to, \textit{e.g.}, ascribe to machines and their actions non-existent morality and common sense, or grow attached to them in ways that cause us to disregard their harms or cloud our judgement of their true capabilities and limitations.}

\AC{To avoid conflation, perhaps machine analogues to terms like `intention' will arise. But wherever the legal and policy language settles, the core philosophical principle -- that a functional interpretation of the `why' of a decision matters for accountability -- will hold. So long as the philosophical (and computational) principles remain, the goals of our research should likewise remain applicable no matter what norms of language develop.}

\section{Technical Background}\label{sec:prelims}

In this section we present some relevant foundations for \acronym{} from formal and automated reasoning.

\subsection{Programs and Traces}\label{subsec:pts}

Let $A$ be the program instantiating a decision algorithm $\mathcal{A}$. $A$ operates over a finite set of program variables $\textsf{var}(A) = V = \{ v_1, \, \ldots, \, v_n \}$. We view $\textsf{var}(A)$ as a union of disjoint subsets $V = I \cup D$. The set $D = \{ vd_1, \, \ldots, \, vd_{n_D} \}$ is the set of internal \emph{decision} variables. The set of input variables $I = E \cup S$ is itself partitioned into sets of \emph{environment} variables $E = \{ ve_1, \, \ldots, \, ve_{n_E} \}$ and \emph{state} variables $S = \{ vs_1, \, \ldots, \, vs_{n_S} \}$. Therefore $n = n_E + n_S + n_D$. $E$ is composed of variables encoding input sources external to the agent, while $S$ is composed of variables encoding internal state. 

Every $v_i$ is associated with a domain $\mathcal{D}_{v_i}$. A \emph{state} is the composition of the variable assignments at that point of the execution $\sigma \in \mathcal{D} = (\mathcal{D}_{v_1} \times \cdots \times \mathcal{D}_{v_n})$. Given $\sigma = (\mathdutchcal d_1, \dots, \mathdutchcal d_n)\in \mathcal D$, we denote the restriction to only environment variables as $\rsigma_E = (\mathdutchcal d_1, \dots, \mathdutchcal d_{e_{n_E}})\in \mathcal{D}_{e_1}\times\cdots\times \mathcal{D}_{e_{n_E}}$, 
and similarly for $\rsigma_S$, $\rsigma_I$, and $\rsigma_D$. A \emph{trace} $\tau = \sigma_1\sigma_2\sigma_3\ldots$ is a (possibly infinite) sequence of states. We access states by $\tau(t) = \sigma_t$, and values of variables at states by $\sigma(v_i) = \mathdutchcal{d}_i$. The set of possible traces is governed by a transition relation $R \subseteq \mathcal{D} \times \mathcal{D}$, so that $\tau(t) = \sigma$ and $\tau(t + 1) = \sigma'$ may occur within some $\tau$ only if $(\sigma, \, \sigma') \in R$. The program $A$ encodes a partial transition relation, $R_{A}$, with the constraint that $(\sigma, \, \sigma') \in R_{A}$ requires that $\forall i \in [n_{E}]. \; \sigma(ve_i) = \sigma'(ve_i)$. That is, by definition $A$ cannot define how the environment $\mathcal{E}[\mathcal{A}]$ updates the $ve_i$, as that capacity is exactly the distinction between a program and the environment it runs within.

We work with statements over the program variables in the logic $\texttt{QF\_FPBV}$. The available domains are those of floating points and bitvectors. An expression $e$ is built from variables composed with the constants and function symbols defined over those domains, \emph{e.g.}, $(\textsf{fp.to\_real } \text{b}011) + 2.34 \cdot v_{14}$. A formula $\varphi$ is built from expressions and relation symbols composed using propositional operators, \emph{e.g.}, the prior expression could extend to the formula $(\textsf{fp.to\_real } \text{b}011) + 2.34 \cdot v_{14} \geq -1.87$.  If $\varphi$ is a formula over $\textsf{var}(A) = V$, notated $\varphi(V)$, and $\sigma = (\mathdutchcal{d}_1, \, \ldots, \, \mathdutchcal{d}_n) \in \mathcal{D}$, then we write $\sigma \models \varphi(V)$ if the constant formula that results from substituting each $\mathdutchcal{d}_i$ for $v_i$ evaluates to $\textsc{True}$. We use $\varphi$ and $\psi$ when writing formulas over $I$ that represent scenarios, $\beta$ when writing formulas over $D$ representing decisions made, and $\Pi$ when writing formulas over $V$ representing whole program executions.

A symbolic state $\hat{\sigma} = \hat{\mathcal{D}} = (\hat{\mathcal{D}}_{\hat{v}_1} \times \cdots \times \hat{\mathcal{D}}_{\hat{v}_n}, \, \pi_{\hat{\sigma}})$ is defined over a set of symbolic variables $\textsf{symvar}(A) = \hat{V} = \{ \hat{v}_1, \, \ldots, \, \hat{v}_n\}$, with $\hat{I}$, $\hat{E}$, $\hat{S}$, and $\hat{D}$ defined analogously. Each $\hat{\mathcal{D}}_{\hat{v}_i}$ augments the concrete domain $\mathcal{D}_{v_i}$ by allowing $\hat{v}_i$ to reference an expression $e_i$ over a set of symbolic values $\{ \alpha_i \}_{i \in [k]}$, \emph{e.g.}, $e_i = \alpha_i$ or $e_i = 2\alpha_j + 3.0$. We write $\hat{\sigma} \models \varphi(\hat{V})$ for formulas over $\textsf{symvar}(A)$ analogously to the concrete case, still in $\texttt{QF\_FPBV}$. The \emph{path constraint} $\pi_{\hat{\sigma}}(\alpha_1, \, \ldots, \, \alpha_n)$ is such a formula over the $\alpha_i$, which captures their possible settings. We let $e_{\hat{\sigma}} \equiv \bigwedge_{i \in [n]} e'_{i}$ where $e'_i \equiv e_i$ for symbolic-valued variables and $e'_i \equiv (\alpha_i = \mathdutchcal{d}_i)$ for concrete-valued variables in $\hat{\sigma}$. Let $\textsf{refs}_{\hat{\sigma}}(\hat{v_i})$ be the $j$ such that $\alpha_j$ is referenced within $e_i$ at state $\hat{\sigma}$. We define the \emph{reference closure} of a symbolic variable recursively, as $\textsf{rcl}_{\hat{\sigma}}(\hat{v}_i) = \{ \alpha_j \}_{j \in \textsf{refs}_{\hat{\sigma}}(\hat{v_i})} \land \bigcup_{k \in \textsf{refs}_{\hat{\sigma}}(\hat{v}_i)} \textsf{rcl}(\hat{v}_k)$. We then define $\left.e_{\hat{\sigma}}\right|_{i}$ as the subformula of $e_{\hat{\sigma}}$ referencing exactly those $\alpha_j \in \textsf{rcl}_{\hat{\sigma}}(\hat{v}_i)$. In sum, $\left.e_{\hat{\sigma}}\right|_{i}$ captures any and all constraints on the value of $\hat{v}_i$ when a symbolic execution of $A$ reaches $\hat{\sigma}$. 

We write $\sigma \models \hat{\sigma}$ if the constant formula that results from substituting each $\mathdutchcal{d}_i$ for $\alpha_i$ in $\pi_{\hat{\sigma}} \land e_{\hat{\sigma}}$ evaluates to $\textsc{True}$. We define the concretization of a symbolic state as $\gamma(\hat{\sigma}) = \{ \sigma \in \mathcal{D} \mid \sigma \models \hat{\sigma} \}$, and the meet $\hat{\sigma}_i \sqcap \hat{\sigma}_j$ by $\gamma(\hat{\sigma}_i \sqcap \hat{\sigma}_j) = \gamma(\hat{\sigma}_i) \cap \gamma(\hat{\sigma}_j)$. A symbolic trace $\hat{\tau} = \hat{\sigma}_1\hat{\sigma}_2\hat{\sigma}_3\ldots$ and a symbolic transition relation $\hat{R}(\hat{\sigma}_i, \, \hat{\sigma}_j)$ are defined as in the concrete case.

We often work with a special class of formulas we call \emph{relaxations}. Relaxations encode formulas that start from a concrete state $\sigma$, and then independently relax the variable constraints so that they can take on a range of values. A \emph{punctured relaxation} excludes a unique model (\emph{i.e.}, state).
\begin{defi}\label{def:relax}
A \textbf{relaxation} is a formula $\varphi(\hat{I}) \equiv \varphi_1(\hat{I}) \land \cdots \land \varphi_n(\hat{I})$ for which every $\varphi_i$ references exactly one $\hat{v}_i$.
\end{defi}
\begin{defi}\label{def:prelax}
A \textbf{punctured relaxation} is a formula $\varphi(\hat{I}) \equiv \varphi'(\hat{I}) \land \lnot F(\hat{I})$ where $\varphi'$ is a relaxation and $\rsigma_{I} \models F(\hat{I})$ for exactly one $\rsigma_{I}$.
\end{defi}

\subsection{SMT-based Program Analysis}

We overview SMT solving and symbolic execution, and refer the reader to~\cite{moura2008z3} and~\cite{cadar2008klee,SurveySymExec-CSUR18} respectively for greater detail.

\paragraph{SMT Solving.} Satisfiability modulo theory (SMT) solving is a form of automated theorem proving that computes the satisfiability (and, by duality, validity) of formulas in certain fragments of first-order (FO) logic. SMT solvers -- we use the state-of-the-art Z3~\cite{moura2008z3} -- are also able to return satisfying models when they exist. In the case of validity queries, these models are concrete counterexamples to the disproven theorem. An SMT formula $\Phi$ is a FO-formula over a decidable theory $T$. In this work, we set $T = \texttt{QF\_FPBV}$, the combination of quantifier-free formulas in the theories of floating-points and bitvectors~\cite{barrett2010smt}. We require support for floating-point statements due to their centrality in machine learning.  n addition to the normal syntax of $\texttt{QF\_FPBV}$, we will need to support a modal counterfactual operator, $x \counterfactual y$, which we formally define in~Section~\ref{subsec:cr}. It is known that modal logics are `robustly decideable' with a generic transformation into a decidable theory~\cite{vardi1997modal}. We however instead use a custom SMT encoding of our counterfactual relation, described in~Section~\ref{sec:rsandqs}. This representation phrases queries so that satisfying models produced by the solver encode concrete counterfactuals.

\paragraph{Symbolic Execution.} One of the great successes of SMT-based program analysis, symbolic execution explores the reachable paths of a program $P$ when executed over $\hat{V}$. Concrete values are computed exactly, assignments to or from symbolic-valued variables update their expressions, and branching conditions update the path constraints. For a branch condition $b(\hat{V})$ reached at symbolic state $\hat{\sigma}_j$, such as a guard for an $\textsf{if}$-statement or $\textsf{while}$-loop, an SMT solver is invoked to check which branches are feasible under $\pi_j$, \emph{i.e.}, whether $\Phi \equiv b(\hat{V}) \land \pi_j$ and/or $\Phi' \equiv \lnot b(\hat{V}) \land \pi_j$ are satisfiable. If only one is, the execution continues along it and its path constraints are updated. For example, if only $\Phi'$ is satisfiable then $\pi_j \leftarrow \pi_{j} \land \lnot b(\hat{V})$. If both are satisfiable, the execution can fork in order to explore all reachable paths and produce a set of constraint formulas $\{\pi_i\}_{i \in [ct]}$ encoding each path at termination. By setting initial constraints on the input variables, symbolic execution can narrow the search space to only the paths of executions meeting preconditions. For our benchmarks we use KLEE-Float~\cite{cadar2008klee,liew2017floating}, a symbolic execution engine for C with support for floating-point operations, backended by Z3. KLEE-Float generates SMT queries in $\texttt{QF\_FPBV}$.\footnote{This requires Ackermannization of arrays to bitvectors~\cite{de2008model}.}

%\paragraph{Running Example.} Suppose under $\varphi''$ of Equation~\ref{eqn:punc}, if it is further the case that $\texttt{agent1\_pos\_x} \leq 1.4$ then $\mathcal{A}$ executes one flow of decision logic, and otherwise it executes another. A symbolic execution of $A$ with inputs constrained as in $\varphi''$ will then produce two paths as output. Suppose in the first path the acceleration guard of $\code{lpr}$ (Figure~\ref{fig:long-pr-example}) is activated but not in the second path. An SMT query as to whether $\code{lprg.active}$ \textit{might} be true over the disjunction of the paths will return true, but if queried whether $\code{lprg.active}$ \textit{would} be true it will return false.

\subsection{Counterfactual Reasoning}\label{subsec:cr}

Counterfactuals are essential to modern theories of causation and responsibility in philosophy and law~\cite{moore2019causation,beebee2019counterfactual,sep-counterfactuals,lewis2013counterfactuals,wachter2017counterfactual}, and are are already quite prominent in formal methods for accountability~\cite{baier2021game, baier2021responsibility, baier2021verification, chockler2004responsibility,feigenbaum2020accountability,datta2015program,halpern2020causes1,halpern2020causes2,wachter2017counterfactual}. Causation refers to the influence an input of some process has on its output, \emph{e.g.}, in an MDP how a choice of action influences the resultant distribution on (some property of) the next state, or for a program how changes in the inputs influence the outputs. In the simplest possible case, an action $a$ is an actual cause of an outcome, such as a harm $h$, if under every counterfactual $a$ is both necessary and sufficient for the harm to occur, notated as $a \counterfactual h$ and $\lnot a \counterfactual \lnot h$. Here $a$ and $h$ are some domain-specific formal objects, while the modal notation $x \counterfactual y$, popularized by Lewis~\cite{lewis2013counterfactuals}, means \emph{if $x$ had happened, then $y$ would have happened}. The canonical unified logical and computational treatment of counterfactuals are the works of Pearl and Halpern~\cite{halpern2020causes1,halpern2020causes2}, which provide a directed acyclic graph-based formalism able to inductively model causal effects in far more complicated dependency structures. 

Responsibility is a higher-order property than causality. As raised in~Section~\ref{sec:motivating}, counterfactuals of decision making are essential to interrogating \emph{intention} and with it responsibility. Put simply, counterfactuals enable challenging and verifying proposed explanations for an agent's decisions. Importantly, counterfactual analysis is well-defined independent of any specific decisions, or the inferences about agent intention made on the basis of them. So although we often frame our discussion in terms of behaviors that match common understanding of human decision making, we stress our formal approach would generalize to future models of algorithmic decision making interested in very different attributes and behaviors than those we apply to humans.

\paragraph{Formalism.} Working \emph{ex post}, we formalize counterfactual algorithmic decision making by starting from a set of \emph{factual} traces $T^f = \{ \tau^f_1, \, \ldots, \, \tau^f_k \}$ encoding a history of $A$'s harmful or otherwise relevant decisions. A \emph{counterfactual} $\tau^{cf} = (\tau^f, \, \tau^{pp}, \, t^*)$ is a tuple of a factual trace $\tau^f$, a \emph{past possibility} trace $\tau^{pp}$, and an integer $t^* \in \mathbb{N}$ that we call the \emph{keyframe}. We write $\tau^{cf}.\textsf{fst} = \tau^f$ and $\tau^{cf}.\textsf{snd} = \tau^{pp}$. Intuitively, we want counterfactuals to represent the decisions that $\mathcal{A}$ \emph{would have} made in revealing alternate circumstances. What makes a counterfactual `revealing' is a deep and nuanced question, but the philosophy of action highlighs the importance of particular attributes for counterfactual scenarios to be meaningful. We enforce these tenets as predicates, in order to guarantee that our method works for counterfactuals possessing them.
\begin{enumerate}[leftmargin=*]
    \itemsep0.8em
    \item \emph{Non-backtracking}: Counterfactuals should encode scenarios with a meaningful relationship to observed events, and should not require us to `replay' the evolution of the world under significant changes to past history. Formally, both $\tau^f$ and $\tau^{pp}$ must be defined at $t^*$, and must agree up until it:
    \begin{equation*}
        \quad \textsf{nbt}(\tau^{cf}) \equiv \forall t'. \, t^* < |\tau^{f}| \land t^* < |\tau^{pp}| \, \land \big(t' < t^* \to \tau^f(t') = \tau^{pp}(t')\big).
    \end{equation*}
    Every past possibility trace forms a non-backtracking counterfactual for $t^* = 1$, so usually choice of keyframe will come first from some \emph{a priori} understanding the investigator has about the critical decision moments leading to a harm. Note that we do not place any restrictions on the (implied) transition $\tau^f(t^* -1) \mapsto \tau^{pp}(t^*)$. In particular, we do not require that $(\tau^f(t^* -1), \, \tau^{pp}(t^*)) \in R_{A}$. Rather, we allow the investigator to `wave a magic wand' in order to define alternate scenarios of interest.\footnote{Counterfactual theories of responsibility often emphasize \emph{similarity} between factual and past possibility~\cite{lewis2013counterfactuals}. In principle, we could model this concern by constraining the distance between $\tau^f(t^*)$ and $\tau^{pp}(t^*)$ by bounding a suitable domain- and program-specific metric $\delta(\tau^f(t^*), \, \tau^{pp}(t^*)) < k_{\delta}$ for constant $k_{\delta}$. As defining metrics able to capture the nuance of large program structures and complex social and physical systems may be extremely difficult, we will instead trust the human expert investigator to know how to specify the `right' scenarios of interest.}
    
    \item \emph{Scope of Decisions}: In order to clarify the purpose of an agent's actions, what an agent \emph{might have done} is less important than what it \emph{would have decided to try to do}. In complex systems the former is often contaminated by the decisions of other agents and the evolution of the environment, as agents rarely have complete control over outcomes. To clarify this distinction, we enforce a scope to the decision making of $\mathcal{A}$ by limiting past possibility traces to internal reasoning. Therefore, no transition after $t^*$ may update the valuations of $E$.
    \begin{equation*}
    \quad \textsf{scope}(\tau, \, t^*) \equiv \forall t' \forall i \in [n_{E}]. \, t^* < t' \to \tau(t' - 1)(ve_i) = \tau(t')(ve_i).
    \end{equation*}
\end{enumerate}
This $\textsf{scope}$ constraint can be interpreted as formalizing that we do not require or use access to an environmental model $\mathcal{E}[\mathcal{A}]$. We can think of this assumption in terms of the transition relation on states that constraints $\tau^{pp}$. We would need $\mathcal{E}[\mathcal{A}]$ in order to define a total transition relation on states $R'_{\mathcal{E}[\mathcal{A}]}(\sigma_i, \, \sigma_{i+1})$, given that $A$ by itself cannot define how the environment changes. By working without an $\mathcal{E}[\mathcal{A}]$, we are limited (enforced by the predicate) to $\tau^{pp}$ produced by the partial relation $R_{A}(\sigma_i, \, \sigma_{i+1})$ only defined on pairs of states that agree on all $ve_i$. Note, however, that this predicate does not preclude $\mathcal{A}$ from incorporating a history of data about $\mathcal{E}[\mathcal{A}]$ into a single decision: the $ve_i$ capture just the freshest environmental information. Any historical data from previous polls of the sensors or network can be passed as various $vs_j$ as $A$ propagates them between executions.

An \emph{admissible} counterfactual is both non-backtracking and limited in scope:
$$\textsf{admit}(\tau^{cf}) \equiv \textsf{nbt}(\tau^{cf}) \land \textsf{scope}(\tau^{pp}, \, t^*).$$
In order to use automated reasoning to interrogate $\mathcal{A}$'s decision making history (in the form of $T^f$), we need to formalize the semantics of two different families of trace properties:
$$\begin{array}{ll}
    \textit{factuals}: & \tau^f \stackrel{?}{\models} \varphi(\hat{I}) \fprop \beta(\hat{D})  \\
    \quad\quad\text{when faced with } \varphi \text{ at time } t \text{, did } \mathcal{A} \text{ do } \beta \text{ at time } \ell \text{?} \span \\[6pt]
    \textit{counterfactuals}: & \tau^{cf} \stackrel{?}{\models} \varphi(\hat{I}) \cprop \beta(\hat{D})  \\
    \quad\quad\text{if faced with } \varphi \text{ at time } t^* \text{, would } \mathcal{A} \text{ have done } \beta \text{ at time } \ell \text{?} \span \\
\end{array}$$
We begin with factuals. In order to formulate a useful semantics for this predicate, we need a reasonable interpretation of the subsequence $\tau^f(t)\ldots\tau^f(\ell)$ that the property implicitly analyzes. Working after-the-fact justifies one: as a \emph{window of agency}, during which either $\mathcal{A}$ made a decision or failed to do so as a harm played out. If the window of agency was still open, we could not be working \emph{ex post}. We can then formulate a semantic definition in which $\varphi(\hat{I})$ specifies preconditions on the inputs to $A$, and $\beta(\hat{D})$ then specifies post-conditions on its decision variables, limited in scope and to the window of agency.
\begin{equation*}
\tau^f \models \varphi(\hat{I}) \fprop \beta(\hat{D}) \quad \quad
\text{ if } \textsf{scope}(\tau^f, \, t) \text{ and } \left.\tau^f(t)\right|_{I} \models \varphi(\hat{I}) \text{ and } \left.\tau^{f}(\ell)\right|_{D} \models \beta(\hat{D}).
\end{equation*}
On the contrary, for counterfactuals it is not obvious that we can assume a known and finite window. As $\tau^{pp}(t^*)$ may have never been observed it could lead to $A$ looping forever, and without an $\mathcal{E}[\mathcal{A}]$ we cannot know how long the window would last. However, as counterfactuals are objects of our own creation, we will assume that the investigator can conjecture a reasonable window $[t^*, \ell]$ within which the decision of $\mathcal{A}$ must be made in order to be timely, with responsibility attaching to the agent if it is unable to make a decision within it. This assumption guarantees termination.

With this philosophically distinct but mathematically equivalent assumption, we are able to define the semantics of the counterfactual operator as
\begin{multline*}
\tau^{cf} \models \varphi(\hat{I}) \cprop \beta(\hat{D}) \quad \quad \text{ if } \textsf{admit}(\tau^{cf}) \text{ and } \left.\tau^f(t^*)\right|_{I} \not\models \varphi(\hat{I}) \text{ and } \hspace{1.6cm} \\ \left.\tau^{pp}(t^*)\right|_{I} \models \varphi(\hat{I}) \text{ and } \left.\tau^{pp}(\ell)\right|_{D} \models \beta(\hat{D}).
\end{multline*}
In practice, our use of symbolic execution will abstract away these details by framing the scope and the window of agency so that $\textsf{admit}(\tau^{cf})$ is true by construction. We discuss this formally in~Section~\ref{sec:rsandqs}.

Lastly, we will oftentimes discuss \emph{families of counterfactuals}, which are sets of counterfactuals which share a factual trace 
$$T^{cf} = \{ \tau_1^{cf}, \, \ldots, \, \tau_{k'}^{cf} \mid \forall i, \, j \in [k']. \, \tau_{i}^{cf}.\textsf{fst} = \tau_{j}^{cf}.\textsf{fst} \}.$$
Families of counterfactuals can naturally be defined implicitly by a tuple $\textsf{ctx} = (\tau^f, \, t^*, \, \varphi)$ as 
\begin{equation*}
T_{\textsf{ctx}}^{cf} = \{ \tau^{cf} \mid \textsf{admit}(\tau^{cf}) \text{ and }  \tau^{cf}.\textsf{fst}(t^*) \not\models \varphi(\hat{I}) \text{ and } \tau^{cf}.\textsf{snd}(t^*) \models \varphi(\hat{I}) \}.
\end{equation*}
Choice of context $\textsf{ctx}$ will be our usual way of delineating families, especially as $\varphi$ then provides a descriptive representation.

\section{Formal Reasoning for Accountability}\label{sec:clear}

Our method, \acronym{}, has both a practical and theoretical basis. Practically, \acronym{} is intended for \emph{counterfactual-guided logic exploration} (\textbf{CLE}AR). Given a program $A$ and log of factual executions $T^f$, \acronym{} provides an interactive, adaptive procedure for the investigator to refine a set \textsc{Facts} of trace properties capturing how $A$ behaves in $T^f$ and related counterfactuals, just as in a legal finding of fact. Theoretically, \acronym{} can also be understood as an instance of semi-automated \emph{counterfactual-guided abstraction refinement} (\textbf{C}LE\textbf{AR}), in the style of the automated CEGAR~\cite{clarke2000counterexample}. This theoretical interpretation provides a potentially fertile ground for future automated extensions of \acronym{} and \soid{} to bolster the explanatory power of the logic exploration.

\subsection{Counterfactual-Guided Logic Exploration}

Given a program $A$ and log of factual executions $T^f$, \acronym{} aims to provide an interactive, adaptive procedure for the investigator to refine a set \textsc{Facts} of trace properties capturing how $A$ behaves in $T^f$ and related counterfactuals, just as in a legal finding of fact. We call this \emph{counterfactual-guided logic exploration}. Our end goal is for its implementation in a tool like \soid{} is to enable continuous refinement of a \emph{formal representation} of $\mathcal{A}$'s decision making:
\begin{equation*}
    \textsc{Facts} = \{ \ldots, \big(\tau_i^f, \, \varphi_{i}(I) \fprop \beta_{i}(V) \big), \ldots \} \, \cup
\{ \ldots, \big(\tau_{j}^{cf}, \, \varphi_{j}(I) \cprop \beta_{j}(V) \big), \, \ldots \}.
\end{equation*}
Each fact in $\textsc{Facts}$ is composed of a (counter)factual trace and a property that holds over it, as proven by an SMT solver. Since we do not assume access to some overarching property $P(A)$ that we aim to prove, $\textsc{Facts}$ is the ultimate product of the counterfactual-guided logic exploration. The human investigator is trusted to take $\textsc{Facts}$ and use it to assess $\mathcal{A}$'s responsibility for a harm.

\RA{Our method relies on an \emph{oracle} interface, $\mathcal{O}_A(\cdot)$, into the decision logic of $A$. We specify factual queries as $q = (\varphi, \, \beta)$, pairing an input constraint $\varphi$ and a behavior $\beta$. Such a query asks whether the factual program execution starting from the program state encoded by $\varphi$ results in the agent behavior encoded by $\beta$, or more formally...
    \begin{quote}
        for $\tau^f$ s.t.~$\tau^{f}(t) \models \varphi(\hat{I})$ uniquely, does $\tau^f \models \varphi(\hat{I}) \fprop \beta(\hat{D})$?
    \end{quote}
    We specify counterfactual queries as $q = (\mathbf{1}_{\exists}, \, \varphi, \, \beta)$, composed of a `might'/`would' (existential/universal) indicator bit $\mathbf{1}_{\exists}$, input constraint $\varphi$, and behavior $\beta$. Each such query asks whether there exists a program execution (a might counterfactual) starting from a program state encoded by $\varphi$ that results in the agent behavior encoded by $\beta$, more formally...
    \begin{quote}
        if $\mathbf{1}_{\exists}$, for $\textsf{ctx} = (\tau^f, \, t^*, \, \varphi)$ does there exist $\tau^{cf} \in T^{cf}_{\textsf{ctx}}$ such that $\tau^{cf} \models \varphi(\hat{I}) \cprop \beta(\hat{D})$;
    \end{quote}
    or similarly, but now whether for all executions starting from the program states encoded by $\varphi$ (a would counterfactual)...
    \begin{quote}
     if $\lnot \mathbf{1}_{\exists}$, for $\textsf{ctx} = (\tau^f, \, t^*, \, \varphi)$, whether for all $\tau^{cf} \in T^{cf}_{\textsf{ctx}}$ it follows that $\tau^{cf} \models \varphi(\hat{I}) \cprop \beta(\hat{D})$?
    \end{quote}
This quite minimal information is sufficient for the oracle to resolve the information needed to improve $\textsc{Facts}$. Note that as $T^{cf}_{\textsf{ctx}}$ excludes the factual trace as a valid continuance from $t^*$, it is not possible for a counterfactual query to resolve (positively or negatively) on the basis of the factual execution -- only counterfactuals are considered.}
\algnewcommand{\IIf}[1]{\State\algorithmicif\ #1}
\algnewcommand{\IThen}[1]{\State\quad\algorithmicthen\ #1}
\algnewcommand{\IElse}[1]{\State\quad\algorithmicelse\hspace{0.42em}\ #1}
\algnewcommand{\ElseIIf}[1]{\algorithmicelse\ #1}
\algnewcommand{\ITE}[2]{\State\quad\algorithmicthen\ #1 \algorithmicelse\ #2}
\begin{algorithm*}[!t]
\caption{\acronym{} Loop}\label{alg:ca}
\begin{algorithmic}[1]
\State $T^{f} \leftarrow \{ \tau_1^{f}, \, \ldots, \, \tau_k^{f} \}$, $\textsc{Facts} \leftarrow \{\}$
\While{$\textit{not done?}$}
\If{$\textit{factual?}$}
\State $(\tau^f, \, t), \, \varphi \leftarrow \textit{start?}(T^{f}, \, I)$
\State $\beta \leftarrow \textit{behavior?}(V)$
\State $r_i = (b, \underline{\hspace{7pt}} \,)\leftarrow \mathcal{O}_{A}(q_i = (\varphi, \, \beta))$
\IIf{$b = 0$}
\IThen{$\textsc{Facts} \leftarrow \textsc{Facts} \cup \{ (\tau^f(t), \, \varphi(\hat{I}) \fprop \lnot \beta(\hat{D})) \}$}
\IElse{$\textsc{Facts} \leftarrow \textsc{Facts} \cup \{ (\tau^f(t), \, \varphi(\hat{I}) \fprop \beta(\hat{D})) \}$}
\Else
\State $\textsf{ctx} = (\tau^f, \, t^*, \, \varphi), \, F \leftarrow \textit{cf}(T^{f}, \, I)$
\State $\beta, \, \mathbf{1}_{\exists} \leftarrow \textit{behavior?}(V)$
\State $r_i = (b, \, \mathcal{M}) \leftarrow \mathcal{O}_{A}(q_i = (\mathbf{1}_{\exists}, \, \varphi, \, \beta, \, F))$
\If{$\mathbf{1}_{\exists} = 0$}
\IIf{$b = 0$}
\IThen{$\textsc{Facts} \leftarrow \textsc{Facts} \cup \{ (\mathcal{M}, \, \varphi(\hat{I}) \cprop \lnot \beta(\hat{D})) \}$}
\IElse{$\textsc{Facts} \leftarrow \textsc{Facts} \cup \{ (\tau_j^{cf}(t^*), \, \varphi(\hat{I}) \cprop \beta(\hat{D})) \}_{\tau_j^{cf} \in T^{cf}_{\textsf{ctx}}}$}
\Else
\IIf{$b = 0$}
\IThen{$\textsc{Facts} \leftarrow \textsc{Facts} \cup \{ (\tau_j^{cf}(t^*), \, \varphi(\hat{I}) \cprop \lnot \beta(\hat{D})) \}_{\tau_j^{cf} \in T^{cf}_{\textsf{ctx}}}$}
\IElse{$\textsc{Facts} \leftarrow \textsc{Facts} \cup \{ (\mathcal{M}, \, \varphi(\hat{I}) \cprop \beta(\hat{D})) \}$}
\EndIf
\EndIf
\EndWhile
\end{algorithmic}
\end{algorithm*}

\paragraph{An Example.} Consider an investigator trying to understand the facts under which the car in Figure~\ref{fig:crash} did, would, or might enter the intersection. If $\varphi$ of Equation~\ref{eqn:fact} represents the critical moment at which the car moved into the intersection, then the investigator could query
$$q_1 = (\varphi, \texttt{move} = 1)$$ 
where $\code{move}$ is a decision variable. If, for example, $r_1 = (1, \_\_\,)$, then Algorithm~\ref{alg:ca} will set
$$\textsc{Facts} = \{ ( \tau^f, \varphi \fprop \texttt{move} = 1) \}$$
to capture the now confirmed fact that the car chose to move into the intersection (rather than say, had a brake failure). Note that to do so the investigator needs only to know the input constraints $\varphi$ and the specific decision variable $\code{move}$. All other aspects of the self-driving car's decision logic is hidden by the oracle interface, and the output is clear and interpretable answer to exactly the question posed. Adaptively, the investigator might then decide to skip Equation~\ref{eqn:cfact}, and instead move on to querying using $\varphi''$ from Equation~\ref{eqn:punc}, \emph{e.g.},
$$q_2 = (\mathbf{1}_{\exists}, \varphi'', \texttt{move} = 0)$$
to ask whether under the family of counterfactuals $T^{cf}_{\textsf{ctx}}$ defined by $\varphi''$ there exists a circumstance where the car would not have entered the intersection. If then, for example, $r_2 = (1, \, \mathcal{M})$, where the model $\mathcal{M}$ encodes a concrete counterfactual scenario, the investigator can update
$$\textsc{Facts} \leftarrow \textsc{Facts} \cup \{ (\mathcal{M}, \varphi'' \cprop \texttt{move} = 0) \}$$
and continue on from there.

\paragraph{The Method.} We define the \acronym{} loop and oracle interface that together underlie \soid{} in Algorithms~\ref{alg:ca} and~\ref{alg:oracle}, where calls in \textit{italics?} indicate manual interventions that must be made by the investigator. The investigatory procedure starts from the set of factual traces $T^f = \{ \tau_1^f, \, \ldots, \, \tau_k^f \}$ observed from $A$'s executions. At each iteration of the loop, the investigator adaptively formulates and poses the next question in a sequence $\textsc{Query} = \langle q_1, \, \ldots, \, q_i, \, \ldots \rangle$. The responses $\textsc{Resp} = \langle r_1, \, \ldots, \, r_i, \, \ldots \rangle$ are then used to build up the set $\textsc{Facts}$ of trace properties regarding $A$'s decision making under both the $T^f$ and the set of counterfactual scenarios, $T^{cf} = \{ \tau_1^{cf}, \, \ldots, \, \tau_{k'}^{cf} \}$, defined within the $q_i$ by the investigator. Each entry in $\textsc{Facts}$ is rigorously proven by the verification oracle $\mathcal{O}_{A}(\cdot)$, with access to the logical representation $\Pi$ of $\mathcal{A}$ as expressed by $A$. We leave to the investigator the decision to terminate the investigatory loop, as well as any final judgement as to the agent's culpability. In~Section~\ref{sec:rsandqs} we explain the encodings $\Phi$ used within Algorithm~\ref{alg:oracle} in detail, and further prove that they correctly implement the semantics of $\fprop$ and $\cprop$ as defined in~Section~\ref{subsec:cr}.

\paragraph{Design Goals.} We briefly highly how  \acronym{} and \soid{} meet some critical design goals to support principled analysis for legal accountability.

\begin{enumerate}
    \itemsep0.3em
    \item The oracle design pushes the technical details of how $\mathcal{A}$ works `across the veil', so that an investigator needs to know no more than the meaning of the input/output API exposed by $A$ (over the variables in $I$ and some subset of $D$, respectively) in order to construct a query $q_i$ and interpret the response $r_i \leftarrow \mathcal{O}_{A}(q_i)$. To this end, we designed the oracle query to place as minimal a possible burden on the investigator.    
    \item The method emphasizes \emph{adaptive} construction of $\textsc{Facts}$, so that the investigator may shape the $i$th query not just by considering the questions $\langle q_1, \, \ldots, \, q_{i-1} \rangle$ asked, but also using the responses $\langle r_1, \, \ldots, \, r_{i-1} \rangle$ already received. We aim to put the agent on the stand, not just send it a questionnaire. Crucial to this goal is to return concrete traces from counterfactual queries, so that their corresponding facts can help guide the construction of the next. Using the ability of SMT solvers to return models for satisfiable formulas, when $\mathbf{1}_{\exists} = 1$ and there exists $\tau^{cf} \in T^{cf}_{\textsf{ctx}}$ such that $\tau^{cf} \models \varphi(\hat{I}) \cprop \beta(\hat{D})$, we are able to explicitly inform the investigator of the fact $(\tau^{cf}, \, \varphi(\hat{I}) \cprop \beta(\hat{D}))$. Conversely, when $\mathbf{1}_{\exists} = 0$ and $\varphi(\hat{I}) \cprop \beta(\hat{D})$ is not true for all $\tau^{cf} \in T^{cf}_{\textsf{ctx}}$, we also can explicitly return the fact $(\overline{\tau^{cf}}, \, \varphi(\hat{I}) \cprop \lnot \beta(\hat{D}))$ for some counterexample $\overline{\tau^{cf}} \in T^{cf}_{\textsf{ctx}}$ encoded by the output model $\mathcal{M}$.    
    \item The method is \emph{interpretable}. When an investigator poses a question, \soid{} pushes everything `smart' the method does across the oracle interface to the verification, so that even non-technical users can understand the relationship between query and response. In a sense, our method benefits from a simple and straightforward design, so that its process is direct and interpretable to the investigators using it. As with a human on the stand, we just want the answer to the question that was asked, no more and no less. As this design goal describes what \emph{not to do} rather than what \emph{to do}, we informally meet it by not introducing unnecessary automation.
\end{enumerate}

\noindent In general, we balance automation against interpretability, in order to minimize the burden on the investigator: we want them to pick a critical moment and (family of) counterfactual scenarios, define a behavior as a post-condition, and get push-button execution. 

\begin{algorithm}[t]
\caption{Oracle}\label{alg:oracle}
\begin{algorithmic}[1]
\Require $q_i = (\varphi, \, \beta) \text{ or } q_i = (\mathbf{1}_{\exists}, \, \varphi, \, \beta, \, F)$
\State $(ct, \, \{ \pi_{ab} \} ) \leftarrow \textsf{SymExec}(A, \, \varphi(\hat{I}))$ \Comment{$ct$ is the branch count}
\State $\Pi \leftarrow \bigvee_{a \in [n], \, b \in [ct]} \pi_{ab}$
\If{$|q| = 2$}
\State $\Phi \leftarrow \lnot \big(\big(\varphi(\hat{I}) \land \Pi(\hat{V})\big) \to \beta(\hat{D})\big)$
\State $(b, \, \mathcal{M}) \leftarrow \textsf{SMT.isValid?}(\Phi)$
\Else
\If{$q.\textsf{fst} = 0$}
\State $\Phi \leftarrow \lnot \big(\big(\varphi(\hat{I}) \land \lnot F(\hat{I}) \land \Pi(\hat{V})\big) \to \beta(\hat{D})\big)$
\State $(b, \, \mathcal{M}) \leftarrow \textsf{SMT.isValid?}(\Phi)$
\Else
\State $\Phi \leftarrow \big(\varphi(\hat{I}) \land \lnot F(\hat{I}) \land \Pi(\hat{V})\big) \, \land$ 
\State \quad\quad\quad\quad\quad\quad\quad\quad $\big(\big(\varphi(\hat{I}) \land \lnot F(\hat{I}) \land \Pi(\hat{V})\big) \to \beta(\hat{D})\big)$
\State{$(b, \, \mathcal{M}) \leftarrow \textsf{SMT.isSat?}(\Phi)$}
\EndIf
\EndIf \\
\Return $r_i = (b, \, \mathcal{M})$
\end{algorithmic}
\end{algorithm}

\subsection{Counterfactual-Guided Abstraction Refinement}\label{subsec:car}

Automated \emph{counterexample-guided abstraction refinement} (CEGAR)~\cite{clarke2000counterexample} is a tentpole of modern formal methods research, underlying many recent advances in program synthesis and analysis. The premise of CEGAR is that instead of trying to prove a safety property $\varphi$ of a program $A$ directly, instead $\varphi$ can be checked against an abstraction $\hat{A}$ that over-approximates the feasible traces of $A$. If the checking of $\varphi$ fails due to a counterexample trace $\overline{\tau} \in \hat{A}$ (\emph{i.e.}, if $\overline{\tau} \not\models \varphi$), the method determines if $\overline{\tau}$ is spurious or real. If real, the property is proven false; if spurious, the abstraction $\hat{A}$ is refined in order to exclude $\overline{\tau}$ explicitly, and the checking restarts. This iterative process of checking and refinement until the property is proven is the \emph{CEGAR loop}. The strength of CEGAR lies in automated discovery of relevance: $\hat{A}$ begins small, and only grows when necessary to accommodate functionality of $A$ that is relevant to its correctness under $\varphi$.

Formally, an abstraction is given by a surjection $h: \mathcal{D} \to \hat{\mathcal{D}}$, where $\hat{\mathcal{D}}$ is an \emph{abstract} domain. Our overloading of the $\hat{\cdot}$ notation from the symbolic domain is intentional. This abstraction function induces an equivalence relation: for $\sigma, \, \sigma' \in \hat{\mathcal{D}}$, $\sigma \equiv \sigma' \text{ iff } h(\sigma) = h(\sigma')$, as well as an over-approximate abstract transition relation over these equivalence classes:
$$(\hat{\sigma}, \, \hat{\sigma}') \in \hat{R} \subseteq \hat{\mathcal{D}} \times \hat{\mathcal{D}} \text{ iff } \exists \sigma, \sigma'\in \mathcal{D}. \; h(\sigma) = \hat{\sigma} \text{ and } h(\sigma') = \hat{\sigma}' \text{ and } (\sigma, \, \sigma') \in R.$$ 
Spurious traces appear because of the equivalence of states produced by a coarse abstraction. If there exists a transition $\sigma \mapsto \sigma'$, and $\sigma \equiv \sigma''$, then $\hat{R}$ includes a transition $\sigma'' \mapsto \sigma'$ within $\hat{A}$ even if it does not exist in $A$ itself. If a spurious counterexample $\overline{\tau}$ includes a transition $\sigma'' \mapsto \sigma'$ along its execution, the refinement process automatically splits $\sigma \not\equiv \sigma''$ to remove this trace.

Like CEGAR, the goal of our \acronym{} is to automatically analyze only relevant executions of a program, in order to efficiently determine its behavior. Unlike in CEGAR, we cannot automate this process completely, because we are unable (by assumption) to assume computable models of either relevance or of correctness. However, the functionality of \acronym{} can be interpreted as providing semi-automated abstraction refinement, based on the (counter)factuals it both analyzes as posed and generates in response to queries.

\begin{figure}[t]
\input{code/pr.tex}
\caption{Longitudinal Proper Response Guard (see Def. 4.1 of~\cite{shalev2017formal}) written in C. Although not used by the simulated self-driving cars of our case study, this sort of decision functionality is representative of what \soid{} is designed to investigate.}
\label{fig:long-pr-example}
\vspace{-7pt}
\end{figure}

Interpreting Algorithm~\ref{alg:ca} as a \emph{counterfactual-guided abstraction refinement} loop follows from viewing the investigator's finding of fact as an iterative refinement of $h$, starting from the trivial abstraction. The investigator begins with no formal understanding of the executions of $A$, \emph{i.e.}, $\textsf{traces}(\textsc{Facts}) = \varnothing$. This implicitly defines an $\hat{A}$ mapping to $\hat{\mathcal{D}} = \{ (\{ \bot_1 \} \times \cdots \times \{ \bot_n \}, \, \textsc{True} ) \} = \{ \hat{\sigma}_{\bot} \}$ where each $\bot_{i}$ is some default value from $\mathcal{D}_i$, and where for all $\sigma \in \mathcal{D}$, $h(\sigma) = \hat{\sigma}_{\bot}$ and $(\hat{\sigma}_{\bot}, \, \hat{\sigma}_{\bot}) \in \hat{R}$. Since this $\hat{A}$ is over a single $\bot$ equivalence class where $\sigma \equiv \sigma'$ for all $\sigma, \, \sigma' \in D$, every trace over $D$ is feasible within the abstracted $\hat{A}$, and so every possible behavior $\beta(\hat{D})$ is the plausible outcome of an execution of $A$. 

Just as when questioning a human, by asking questions and receiving answers -- the \acronym{} loop -- the investigator iteratively clarifies their understanding of the agent's decision making. As a high-level sketch, consider an investigation into the behavior of a program $A$ invoking the autonomous vehicle control code in Figure~\ref{fig:long-pr-example}, as shown in Figure~\ref{fig:CLEAR}. When the investigator poses a counterfactual query $q_i$, the resulting symbolic execution follows a set of symbolic transitions $\hat{R}_{q_i} \subseteq \hat{\mathcal{D}} \times \hat{\mathcal{D}}$ to visit a set of symbolic states $\{ \hat{\sigma_{t^*}}, \, \ldots, \, \hat{\sigma}_{\ell} \}$. The same is true of factual queries, albeit where each $\hat{\sigma_j}$ is equivalent to a corresponding concrete state $\sigma_j$, \emph{i.e.}, $\sigma_j \models \hat{\sigma_j}$ uniquely. In the figure we use such symbolic and concrete states interchangeably. $\hat{R}_{q_i}$ and the $\hat{\sigma_j}$ can be used to refine $\hat{A}$. If a symbolic state is distinct from $\hat{A}$, that is, if $\hat{\sigma_j} \sqcap \hat{\sigma}'$ for all $\hat{\sigma}' \in \hat{\mathcal{D}}$, then $\hat{\mathcal{D}} \leftarrow \hat{\mathcal{D}} \cup \{ \hat{\sigma_j} \}$ directly. When the meet is not empty, a new symbolic state can be generated by conjoining the variable expressions and path constraints of $\hat{\sigma_j}$ and $\hat{\sigma}'$ that concretizes to $\gamma(\hat{\sigma_j}) \cap \gamma(\hat{\sigma}')$. The symbolic transitions in $\hat{R}_{q_i}$ can then be added to $\hat{R}$ as appropriate to maintain the over-approximation.

At the moment, the observation that the product of iterative symbolic executions can be related to abstraction refinement -- as has been previously observed in the other direction, see \emph{e.g.},~\cite{beyer2016symbolic} -- is purely theoretical for \acronym{}. Nonetheless, we find this abstraction refinement perspective to have great potential. In normal CEGAR, the trace and trace formulas used during refinement are spurious -- instead of being added to the abstraction directly, the abstraction is refined by using them to find predicates that exclude them. Finding the optimal refinement is NP-hard~\cite{clarke2000counterexample}, so numerous methods have been developed to efficiently and effectively extract powerful predicates that capture relevant program behavior~\cite{jhala2009software}. For example, prior work combining CEGAR and symbolic execution applies interpolation methods~\cite{beyer2016symbolic}. A potential direction for future work might be to determine whether such algorithms can be adapted to automatically produce higher-level abstract representations of critical program behavior useful to \acronym{}. Such an extension might be able to, \emph{e.g.}, synthesize counterfactual families that capture a set of concrete states recurring across a number of queries, in order to suggest a counterfactual generation query primed to produce a particularly informative example. By using abstraction refinement, we believe \acronym{} might be enhanced for better explanations without unbalancing the interpretability vs.~automation tradeoff.

\begin{figure}[!t]
\centering
\includegraphics[width=\textwidth]{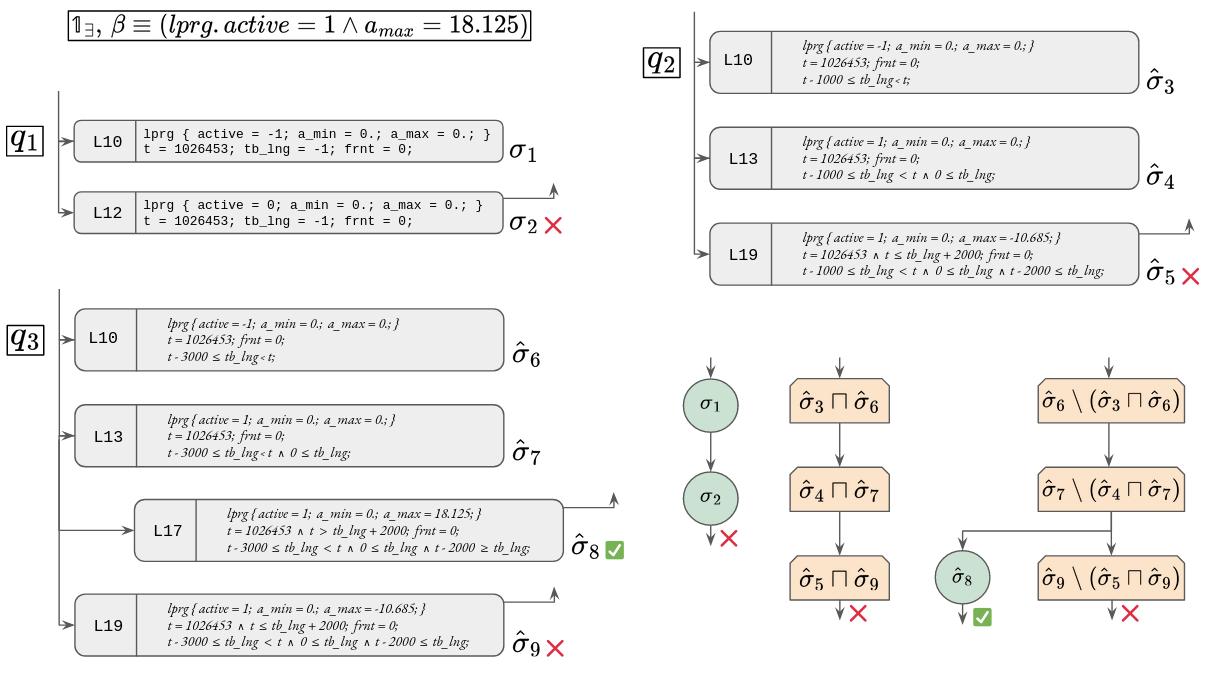}
\caption{Counterfactual-guided Abstraction Refinement of the code in Figure~\ref{fig:long-pr-example}. A factual query $q_1$ and a pair of counterfactual queries, $q_2, \, q_3$, all with the same $\beta$, produce the diagrammed symbolic traces and the resultant abstraction $\hat{A}$ at bottom right. Concrete states are rendered with a {\tt typewriter} font, while symbolic states are rendered in {\garamond Garamond}.}
\label{fig:CLEAR}
\vspace{-7pt}
\end{figure}

\paragraph{An Example.} Figure~\ref{fig:CLEAR} shows the evolution of a counterfactual-guided abstraction refinement over our running example inspecting $\code{lpr}$ from Figure~\ref{fig:long-pr-example}. Through a sequence of three queries $\langle q_1, \, q_2, \, q_3 \rangle$, where the first are factual and the latter two counterfactual, the investigator attempts to find a concrete counterfactual where the acceleration guard will not only be activated, but set to a particular value. As the queries are resolved, the symbolic executions reach nine symbolic states, all but two of which can be concretized into multiple concrete states. The meets of some of these states are non-empty, \emph{e.g.}, $\hat{\sigma}_3 \sqcap \hat{\sigma}_6 \neq \varnothing$, as the conditions of $\hat{\sigma}_6$ imply those of $\hat{\sigma}_3$, so $\gamma(\hat{\sigma}_3) \subseteq \gamma(\hat{\sigma}_6)$. Those meets form distinct symbolic states within the abstraction, with the transition relation capturing the possible executions appropriately.  

\section{Representations and Queries}\label{sec:rsandqs}

We specify a factual query as a tuple $(\varphi, \beta)$, the former logical formula specifying the inputs to the program at a critical decision moment (as in Equation~\ref{eqn:fact}), the latter encoding a description of the possible agent decision being investigated. Implicitly, $\varphi$ defines a factual scenario $(\tau^f, \, t, \, \varphi)$, where $\varphi$ encodes the program state at that critical moment $\tau^f(t)$. Counterfactual queries are encoded similarly, but with the additional of the existential indicator bit $\mathbf{1}_{\exists}$. They are also able to encode many different possible program executions, captured by the notion of a family of counterfactuals $T^{cf}_{\textsf{ctx}}$. The last necessary statement required to invoke an SMT solver is $\Pi(\hat{V})$, the decision logic of $A$ constrained to the scenario(s) implied by $\varphi$. We generate $\Pi(\hat{V})$ dynamically given a (counter)factual query using symbolic execution.

\subsection{Representing Agents and Scenarios}\label{subsec:encodings}

We represent scenarios by tight constraints on the variables in $\hat{I}$ as encoded by $\tau^f(t)$, and counterfactual scenarios as punctured relaxations of $\tau^f(t^*)$. With this restriction, in order to distinguish constraints over $E$ and $S$ we are able to redefine $\varphi(\hat{I}) = \varphi(\hat{E}) \land \psi(\hat{S})$. Working solely with relaxations for specifiable counterfactuals is an ergonomic tradeoff to limit the expressive power of queries, in order to promote adaptive sequential querying of statements amenable to symbolic execution. In principle, \acronym{} supports supports counterfactuals written as arbitrary $\varphi(\hat{I}) \in \texttt{QF\_FPBV}$, and $\soid{}$ could do so as well.  

\paragraph{Factuals.} Factual scenarios are encoded by an (admittedly paradoxical) `tight relaxation', in the sense that every input variable is constrained by an equality.

\begin{defi}\label{def:fact}
A \textbf{factual scenario} is a tuple $(\tau^f, \, t, \, \varphi, \, \psi)$ such that
\begin{enumerate}[i)]
    \item $\varphi(\hat{E}) \land \psi(\hat{S})$ is a relaxation;
    \item $\left.\tau^f(t)\right|_{I} \models \varphi(\hat{E}) \land \psi(\hat{S})$; and
    \item for all $\rsigma_{I} \neq \left.\tau^f(t)\right|_{I}$, $\rsigma_{I} \not\models \varphi(\hat{E}) \land \psi(\hat{S})$.
\end{enumerate}
\end{defi}

\noindent In practice, a factual query is specified by a $(\varphi, \, \psi)$ pair that have a unique satisfying model over $I$. Evaluating a factual scenario is functionally equivalent to a concrete execution, since $\tau^f$ is the only possible program trace. We tie factual analysis into our framework for completeness, and because unlike traditional `opaque' assertion-based testing $\soid{}$ supports writing complex behavioral conditions on all of $V$, including both internal and output variables. Additionally, our factual representations also naturally generate circuits for (zero-knowledge) proofs-of-compliance, another promising tool for algorithmic accountability~\cite{kroll2017accountable,ozdemir2022circ}.

\paragraph{Counterfactuals.} A counterfactual is encoded as a punctured relaxation which removes the original factual $\tau^{f}$ as a valid model, $\varphi(\hat{I}) \equiv \big(\varphi(\hat{E}) \land \psi(\hat{S}) \big) \land \lnot F(\hat{I})$. 

\begin{defi}\label{def:cfact}
\begin{itemize}
    \item[] % DIRTY HACK
    \item[1.] A \textbf{counterfactual scenario} is a tuple $(\tau^{cf} = (\tau^f, \, \tau^{pp}, \, t^*), \, \varphi, \, \psi, \, F)$ such that 
    \begin{enumerate}[i)]
        \item $\big(\varphi(\hat{E}) \land \psi(\hat{S})\big) \land \lnot F(\hat{I})$ forms a punctured relaxation;
        \item $\left.\tau^{pp}(t^*)\right|_{I} \models \varphi(\hat{E}) \land \psi(\hat{S}) \land \lnot F(\hat{I})$;
        \item $\left.\tau^f(t^*)\right|_{I} \models F(\hat{I})$; and
    \end{enumerate}
    \item[2.] A \textbf{family of counterfactual scenarios} is a tuple $(T^{cf}_{\textsf{ctx}}, \, \varphi, \, \psi, \, F)$ where the set $T^{cf}_{\textsf{ctx}}$ contains every $\tau^{cf} = (\tau^f, \, \tau^{pp}, \, t^*)$ such that $\left.\tau^{pp}\right|_{I} \models \varphi(\hat{E}) \land \psi(\hat{S}) \land \lnot F(\hat{I})$ and $(\tau^{cf}, \, \varphi, \, \psi, \, F)$ is a counterfactual scenario.
\end{itemize}
\end{defi}
\noindent In practice, a counterfactual query is specified by a $(\mathbf{1}_{\exists}, \, \varphi, \, \psi, \, F)$ tuple where $\tau^{f}(t^*)$ is excluded as a model by the negation of the formula $F(\hat{I})$ tightly encoding it.

\paragraph{Behaviors.} A behavior is just an arbitrary formula over $\hat{D}$.

\begin{defi}\label{def:behav}
A \textbf{behavior} is a formula $\beta(\hat{D})$.
\end{defi}

\noindent Let $\textsf{vars}(\beta) \subseteq [n]$ be the set of indices such that $i \in \textsf{vars}(\beta)$ if $\hat{v}_i$ is referenced by $\beta(\hat{D})$. We then define $\left.e_{\hat{\sigma}}\right|_\beta \equiv \bigwedge_{i \in \textsf{vars}(\beta)} \left.e_{\hat{\sigma}}\right|_i$ as the formula describing the assignment constraints of all variables in $\beta$ at a given symbolic state $\hat{\sigma}$.

\paragraph{Decision Logic.} The decision logic of the agent is composed of two sets of interrelated components. The first set of components are the path formulas $\pi_i$ for $i \in [ct]$ generated by a symbolic execution, when constrained by $\varphi(\hat{E}) \land \varphi(\hat{S})$ and up to some maximum step length $\ell_{\textsf{max}}$. This `maximum time' guarantees each symbolic trace $\hat{\tau}_i$ of the symbolic execution will terminate at some state $\hat{\sigma}_{\ell_i}$ for $\ell_i \leq \ell_{\textsf{max}}$. The second set are the $\left.e_{\hat{\sigma}_{\ell_i}}\right|_{\beta}$ describing the assignment constraints for variables in $\beta$ at those terminating states. Composed together, these components produce a formula $\Pi(\hat{V})$ that constrains the possible output values of the variables in $\beta$ given $\varphi(\hat{E}) \land \varphi(\hat{S})$ as preconditions.\footnote{For counterfactual queries we could in theory incorporate $\lnot F(\hat{I})$ into the symbolic execution preconditions, in order to possibly preclude the path containing $\tau^f$ from $\{ \pi_{i}(\hat{V}) \}_{i \in [ct]}$. In practice this only excludes a single path only some of the time and does not alter correctness, so we omit it from our definition.}

\begin{defi}\label{def:agent}
The \textbf{decision logic} of $A$ is a formula 
$$\Pi_{\varphi, \psi}^{\beta, \ell_{\textsf{max}}}(\hat{V}) \equiv \bigvee_{i \in [ct]} \pi_i(\hat{V}) \land \left.e_{\hat{\sigma}_{\ell_i}}\right|_{\beta}$$
where $\{\pi_i(\hat{V})\}_{i \in [ct]}$ is the set of path constraints produced by a symbolic execution of $A$ with $\pi_1 \equiv \varphi(\hat{E}) \land \psi(\hat{S})$ as the initial path constraint and $\ell_{\textsf{max}}$ as the timeout, and where $\hat{\sigma}_{\ell_i}$ is the terminating state of the $i$-th symbolic trace of the execution.
\end{defi}

\noindent For concision, we continue to write $\Pi(\hat{V})$. A symbolic execution engine like KLEE-Float can be coaxed to automatically append each $\left.e_{\hat{\sigma}_{\ell_i}}\right|_{\beta}$ to $\pi_i$ by assuming every $\hat{v}_i$ referenced in $\beta$ is equal to a fresh symbolic variable right after $\hat{v}_i$ is (symbolically) computed. 

\subsection{Resolving (Counter)factual Queries}

Given these representations, it is possible for us to encode the semantics of our factual $(\fprop)$ and counterfactual $(\cprop)$ operators as SMT queries in $\texttt{QF\_FPBV}$. In the following arguments, we assume correctness of symbolic execution, \emph{i.e.}, that $\Pi(\hat{V})$ exactly represents the possible executions of $A$ under $\varphi(\hat{E}) \land \psi(\hat{S})$ up to some time $\ell \leq \ell_{\textsf{max}}$.

\begin{thm}\label{thm:fact}
Let $q_i = ((\varphi, \psi), \beta)$ be a factual query, and let $(\tau^f, \, t, \, \varphi, \, \psi)$ be a corresponding factual scenario. Then
$$\Phi \equiv \big(\varphi(\hat{E}) \land \psi(\hat{S}) \land \Pi(\hat{V})\big) \to \beta(\hat{D})$$
is valid iff $\tau^f \models \varphi(\hat{E}) \land \psi(\hat{S}) \fprop \beta(\hat{D})$.
\end{thm}

\begin{proof}
We consider each direction in turn.

\vspace{3mm}

\noindent$\Rightarrow:$ By assumption, \textbf{Definition~\ref{def:fact}(ii)}, and \textbf{Definition~\ref{def:fact}(iii)}, $\left.\tau^f(t)\right|_I \models \varphi(\hat{E}) \land \psi(\hat{S})$ uniquely. Therefore $\Phi$ is true for any valuation $\mathcal{M}$ of $\hat{V}$ for which $\left.\mathcal{M}\right|_I \neq \left.\tau^f(t)\right|_I$, as the LHS of the implication is false. So, the validity of $\Phi$ reduces to its truth when $\left.\mathcal{M}\right|_I = \left.\tau^f(t)\right|_I$.

By assumption $\textsf{scope}(\tau^f, \, t)$ is true, implying $\left.\tau^{f}(\ell)\right|_{I} = \left.\tau^{f}(t)\right|_{I}$. By the correctness of the symbolic execution $\Pi(\hat{V})$ is satisfiable if and only if $\varphi({\hat{E}}) \land \psi(\hat{S})$ is as well, and in particular $\tau^{f}(\ell) \models \Pi(\hat{V})$ uniquely for some $\ell \leq \ell_{\textsf{max}}$. Therefore for any valuation $\mathcal{M} \neq \tau^f(\ell)$ it follows that $\mathcal{M} \not\models \Pi(\hat{V})$, and so $\Phi$ is true as the LHS of the implication is again false. That leaves only the case of $\tau^{f}(\ell)$. As we have already shown $\tau^{f}(\ell) \models \varphi(\hat{E}) \land \psi(\hat{S}) \land \Pi(\hat{V})$, and as by assumption further $\left.\tau^{f}(\ell)\right|_{D} \models \beta(\hat{D})$, we may conclude that $\tau^{f}(\ell) \models \Phi$.

\vspace{3mm}

\noindent$\Leftarrow$: By the validity of $\Phi$, $\tau^f(\ell') = \mathcal{M}' \models \Phi$ for every $\ell'$. First, by the correctness of the symbolic execution, $\left.\tau^f(\ell')\right|_{I} = \left.\tau^f(t)\right|_{I}$ always, implying $\textsf{scope}(\tau^f, \, t)$ is true. We show there exists some $\ell$ such that $\tau^f(\ell) = \mathcal{M} \models \varphi(\hat{E}) \land \psi(\hat{S}) \land \Pi(\hat{V})$, from which it follows that $\left.\tau^{f}(\ell)\right|_{D} \models \beta(\hat{D})$, allowing us to conclude that $\tau^f \models \varphi(\hat{E}) \land \psi(\hat{S}) \fprop \beta(\hat{D})$.

That there exists some unique $\ell \leq \ell_{\textsf{max}}$ such that $\tau^{pp}(\ell) = \mathcal{M} \models \Pi(\hat{V})$ follows from the correctness of the symbolic execution. Finally, since $\left.\tau^{f}(\ell)\right|_{I} = \left.\tau^{f}(t)\right|_{I}$ by \textbf{Definition~\ref{def:fact}(ii)} we have $\left.\tau^{f}(t)\right|_{I} = \left.\tau^{f}(\ell)\right|_{I} = \left.\mathcal{M}\right|_{I} \models \varphi(\hat{E}) \land \psi(\hat{S})$. 
\end{proof}

\noindent Notice that our proof is unaltered by dropping property \textbf{Definition~\ref{def:fact}(i)} from the definition, since as noted the restriction to relaxations is an ergonomic decision.

\begin{thm}\label{thm:cfact}
Let $q_i = (\mathbf{1}_{\exists}, (\varphi, \psi), \beta)$ be a counterfactual query, and let $(T^{cf}_{\textsf{ctx}}, \varphi, \, \psi, \, F)$ be a corresponding family of counterfactual scenarios. Then
\begin{enumerate}[i)]
    \item for $\lnot\mathbf{1}_{\exists}$
    $$\Phi \equiv \big(\varphi(\hat{E}) \land \psi(\hat{S}) \land \lnot F(\hat{I}) \land \Pi(\hat{V})\big) \to \beta(\hat{D})$$
    is valid iff $\forall \tau^{cf} \in T^{cf}_{\textsf{ctx}}. \,\, \tau^{cf} \models \varphi(\hat{E}) \land \psi(\hat{S}) \land \lnot F(\hat{I}) \cprop \beta(\hat{D})$; and
    \item for $\mathbf{1}_{\exists}$
    $$\Phi \equiv \big(\varphi(\hat{E}) \land \psi(\hat{S}) \land \lnot F(\hat{I}) \land \Pi(\hat{V})\big) \land \\ \big(\big(\varphi(\hat{E}) \land \psi(\hat{S}) \land \lnot F(\hat{I}) \land \Pi(\hat{V})\big) \to \beta(\hat{D})\big)$$
    is satisfiable iff $\exists \tau^{cf} \in T^{cf}_{\textsf{ctx}}. \,\, \tau^{cf} \models \varphi(\hat{E}) \land \psi(\hat{S}) \land \lnot F(\hat{I}) \cprop \beta(\hat{D})$.
\end{enumerate}
\end{thm}
\begin{proof}
We consider each direction in turn for both query types.

\vspace{3mm}

\noindent$\lnot\mathbf{1}_{\exists}, \, \Rightarrow:$ Let $\tau^{cf} \in T^{cf}_{\textsf{ctx}}$. Recall $\tau^{cf} = (\tau^f, \tau^{pp}, t^*)$. By assumption and \textbf{Definition~\ref{def:cfact}.1(ii)}, $\left.\tau^{pp}(t^*)\right|_I \models \varphi(\hat{E}) \land \psi(\hat{S}) \land \lnot F(\hat{I})$, and $T^{cf}_{\textsf{ctx}}$ contains every such $\tau^{cf}$ by $\textbf{Definition~\ref{def:cfact}.2}$. Therefore $\Phi$ is true for any valuation $\mathcal{M}$ of $\hat{V}$ for which $\left.\mathcal{M}\right|_I \neq \left.\tau^{pp}(t^*)\right|_I$ for every $\tau^{cf} \in T^{cf}_{\textsf{ctx}}$, as the LHS of the implication is false. Moreover \textbf{Definition~\ref{def:cfact}.1(iii)} implies that $\left.\tau^{f}(t^*)\right|_I \not\models F(\hat{I})$, and so for any $\mathcal{M}$ for which $\left.\mathcal{M}\right|_I = \left.\tau^{f}(t^*)\right|_I$ the LHS of the implication is again false. Together, the validity of $\Phi$ reduces to its truth in the case that $\left.\mathcal{M}\right|_I = \left.\tau^{pp}(t^*)\right|_I \neq \left.\tau^{f}(t^*)\right|_I$ for some $\tau^{cf} \in T^{cf}_{\textsf{ctx}}$.

By the correctness of the symbolic execution $\Pi(\hat{V})$ is satisfiable if and only if $\varphi({\hat{E}}) \land \psi(\hat{S})$ is as well, and in particular $\tau^{pp}(\ell) \models \Pi(\hat{V})$ for some $\ell \leq \ell_{\textsf{max}}$ for every $\tau^{cf} \in T^{cf}_{\textsf{ctx}}$, or $\tau^{f}(\ell) \models \Pi(\hat{V})$. If the latter, $\tau^f(\ell) \not\models \lnot F(\hat{I})$ by \textbf{Definition~\ref{def:cfact}.1(i)} and \textbf{Definition~\ref{def:prelax}}, and so $\Phi$ is true as the LHS of the implication is again false. Otherwise, if instead the former, we split into three cases, where $\tau^{cf}, \, \tau'^{cf} \in T^{cf}_{\textsf{ctx}}$ are arbitrary but distinct counterfactuals.

\begin{enumerate}[a)]
    \item if WLOG $\left.\mathcal{M}\right|_{I} \neq \left.\tau^{pp}(\ell)\right|_{I}$ or $\left.\mathcal{M}\right|_{D} \neq \left.\tau'^{pp}(\ell')\right|_{D}$, it follows that $\mathcal{M} \not\models \Pi(\hat{V})$, and so $\Phi$ is true as the LHS of the implication is again false.
    \item By assumption that $\textsf{scope}(\tau^{pp}, \, t^*)$ and $\textsf{scope}(\tau'^{pp}, \, t^*)$ are true and $\tau^{cf}$ and $\tau'^{cf}$ are distinct, $\left.\tau^{pp}(\ell)\right|_{I} = \left.\tau^{pp}(t^*)\right|_{I} \neq \left.\tau'^{pp}(t^*)\right|_{I} = \left.\tau'^{pp}(\ell')\right|_{I}$. As such, WLOG $\left.\mathcal{M}\right|_{I} = \left.\tau^{pp}(\ell)\right|_{I}$ and $\left.\mathcal{M}\right|_{D} = \left.\tau'^{pp}(\ell')\right|_{D}$ means $\mathcal{M} \not\models \Pi(\hat{V})$ by the correctness of the symbolic execution, and so $\Phi$ is true as the LHS of the implication is again false.
    \item if $\mathcal{M} = \tau^{pp}(\ell)$, then as by assumption $\left.\tau^{pp}(\ell)\right|_{D} \models \beta(\hat{D})$, we may conclude that $\tau^{pp}(\ell) \models \Phi$.
\end{enumerate}

\vspace{3mm}

\noindent$\lnot\mathbf{1}_{\exists}, \, \Leftarrow$: By the validity of $\Phi$, $\tau^{pp}(\ell') = \mathcal{M}' \models \Phi$ for every $\ell'$ for every $\tau^{cf} \in T^{cf}_{\textsf{ctx}}$. First, by the correctness of the symbolic execution, $\left.\tau^{pp}(\ell')\right|_{I} = \left.\tau^{pp}(t^*)\right|_{I}$ always, implying $\textsf{scope}(\tau^{pp}, \, t^*)$ is true. Second, as $\Phi$ does not constrain any $t' < t^*$, $\textsf{nbt}(\tau^{cf})$ follows trivially, and so $\textsf{admit}(\tau^{cf})$ is true as well. We show that for every $\tau^{cf} \in T^{cf}_{\textsf{ctx}}$ there exists some $\ell$ such that $\tau^{pp}(\ell) = \mathcal{M} \models \varphi(\hat{E}) \land \psi(\hat{S}) \land \lnot F(\hat{I}) \land \Pi(\hat{V})$, from which it follows that $\left.\tau^{pp}(\ell)\right|_{D} \models \beta(\hat{D})$, allowing us to conclude that $\tau^{cf} \models \varphi(\hat{E}) \land \psi(\hat{S}) \land \lnot F(\hat{I}) \cprop \beta(\hat{D})$.

That there exists some unique $\ell \leq \ell_{\textsf{max}}$ such that $\tau^{pp}(\ell) = \mathcal{M} \models \Pi(\hat{V})$ follows from the correctness of the symbolic execution. By \textbf{Definition~\ref{def:cfact}(i)}, \textbf{Definition~\ref{def:cfact}(iii)}, and \textbf{Definition~\ref{def:prelax}} we have $\left.\tau^{f}(t^*)\right|_{I} \models F(\hat{I})$ uniquely, and so $\left.\tau^{f}(t^*)\right|_{I} \not\models \varphi(\hat{E}) \land \psi(\hat{S}) \land \lnot F(\hat{I})$. Finally, since $\left.\tau^{pp}(\ell)\right|_{I} = \left.\tau^{pp}(t^*)\right|_{I}$ by \textbf{Definition~\ref{def:cfact}(ii)} we have $\left.\tau^{pp}(t^*)\right|_{I} = \left.\tau^{pp}(\ell)\right|_{I} = \left.\mathcal{M}\right|_{I} \models \varphi(\hat{E}) \land \psi(\hat{S}) \land \lnot F(\hat{I})$.

\vspace{3mm} 

\noindent$\mathbf{1}_{\exists}, \, \Rightarrow$: Let $\tau^{cf} \in T^{cf}_{\textsf{ctx}}$. Recall $\tau^{cf} = (\tau^f, \tau^{pp}, t^*)$. By assumption and \textbf{Definition~\ref{def:cfact}.1(ii)}, $\left.\tau^{pp}(t^*)\right|_I \models \varphi(\hat{E}) \land \psi(\hat{S}) \land \lnot F(\hat{I})$, and that $\textsf{scope}(\tau^{pp}, \, t^*)$ is true, implying $\left.\tau^{pp}(\ell)\right|_{I} = \left.\tau^{pp}(t^*)\right|_{I}$. By the correctness of the symbolic execution $\Pi(\hat{V})$ is satisfiable if and only if $\varphi({\hat{E}}) \land \psi(\hat{S})$ is as well, and in particular $\tau^{pp}(\ell) \models \Pi(\hat{V})$ for some $\ell \leq \ell_{\textsf{max}}$. Moreover, by assumption $\left.\tau^{pp}(\ell)\right|_{D} \models \beta(\hat{D})$, and so we may conclude $\tau^{pp}(\ell) \models \Phi$.

\vspace{3mm}

\noindent$\mathbf{1}_{\exists}, \, \Leftarrow$: Let $\mathcal{M} \models \Phi$. As $\left.\mathcal{M}\right|_{I} \models \varphi(\hat{E}) \land \psi(\hat{S}) \land \lnot F(\hat{I}) \land \Pi(\hat{V})$, by $\textbf{Definition~\ref{def:cfact}.2}$ and the correctness of symbolic execution it follows that $\mathcal{M}$ encodes a $\tau^{pp}(\ell)$ for some $\tau^{cf} = (\tau^f, \tau^{pp}, t^*) \in T^{cf}_{\textsf{ctx}}$  for some $\ell \leq \ell_{\textsf{max}}$. Moreover, by the correctness of the symbolic execution, $\left.\tau^{pp}(\ell)\right|_{I} = \left.\tau^{pp}(t^*)\right|_{I}$, implying $\textsf{scope}(\tau^{pp}, \, t^*)$ is true. As $\Phi$ does not constrain any $t' < t^*$, $\textsf{nbt}(\tau^{cf})$ follows trivially, and so $\textsf{admit}(\tau^{cf})$ is true as well. Finally, by assumption, $\left.\mathcal{M}\right|_{D} \models \beta(\hat{D})$ and $\left.\mathcal{M}\right|_{I} \not\models F(\hat{I})$, allowing us to conclude that $\tau^{cf} \models \varphi(\hat{E}) \land \psi(\hat{S}) \land \lnot F(\hat{I}) \cprop \beta(\hat{D})$.
\end{proof}
\noindent Similarly, our proof is unaltered by dropping property \textbf{Definition~\ref{def:cfact}(i)} from the definition.

\begin{figure*}[!t]
    \centering
    \begin{subfigure}[c]{0.9\textwidth}
        \centering
        \includegraphics[width=0.80\textwidth]{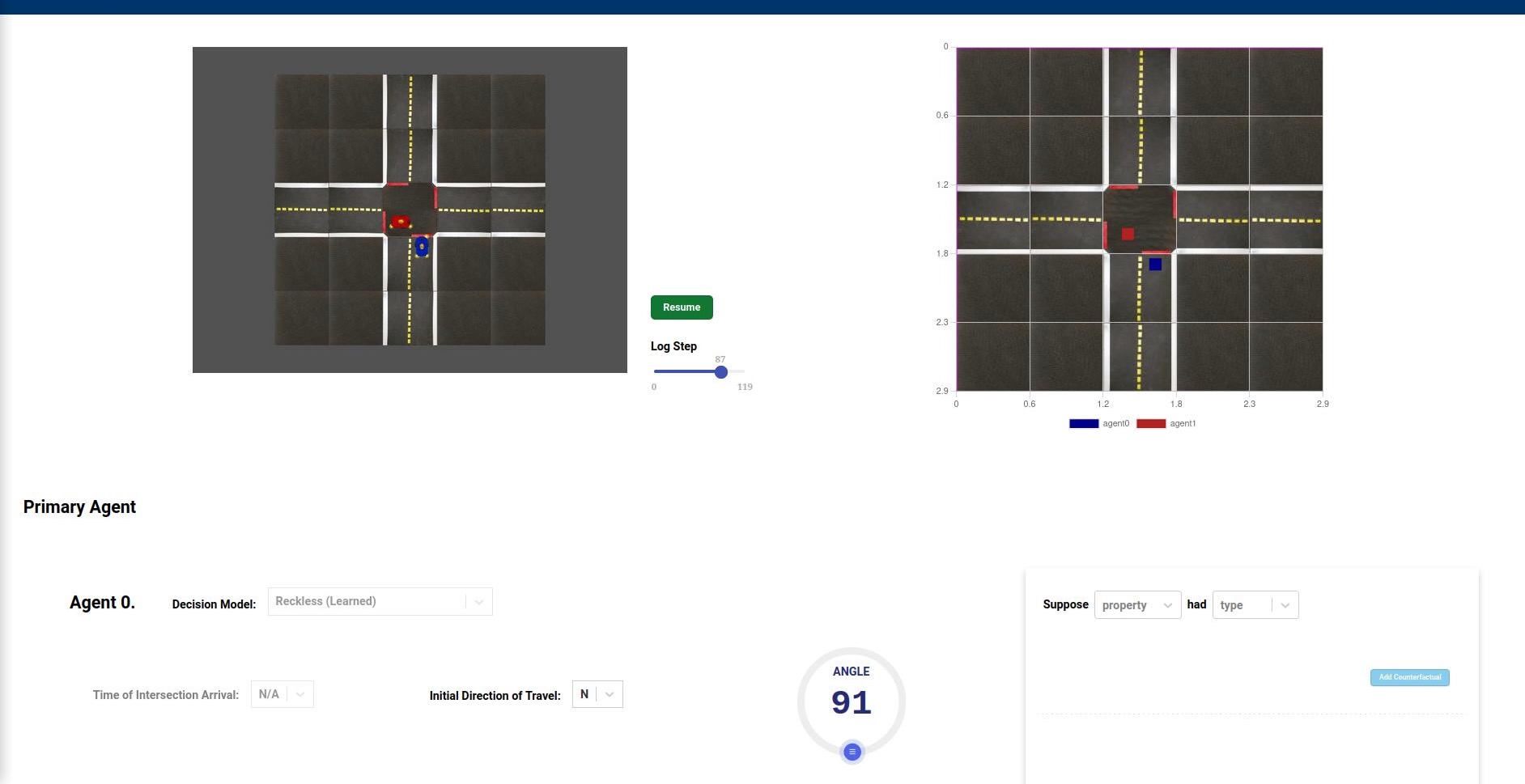}
    \end{subfigure}
    \vspace{6mm}
    \hspace{0.0mm}
    \begin{subfigure}[c]{0.9\textwidth}
        \centering
        \includegraphics[width=0.80\textwidth,trim={0cm 4cm 0cm 0cm},clip]{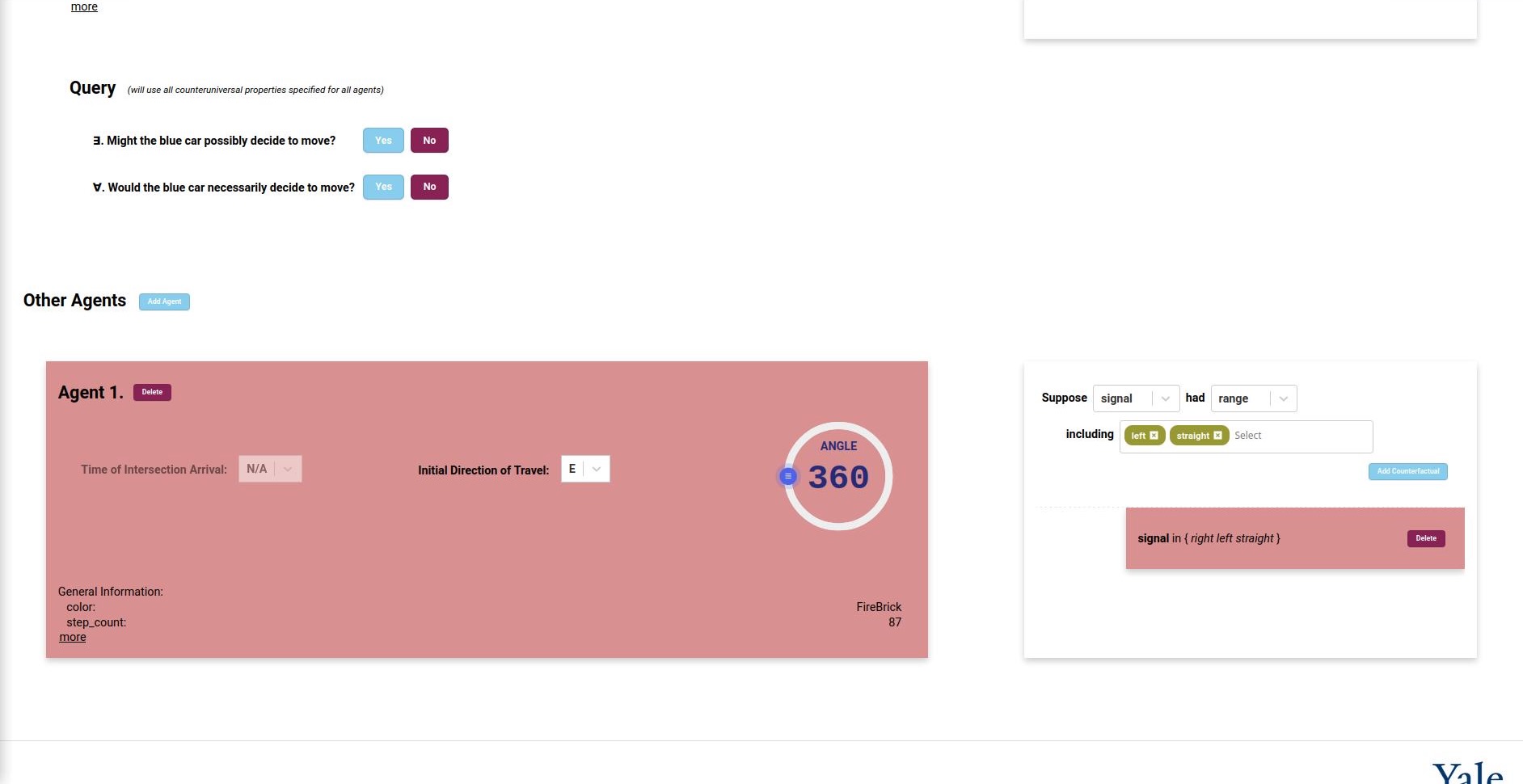}
    \end{subfigure}
    \vspace{-5mm}
    \caption{Still of the \soid{} GUI (with a small section cut out for brevity). At top left is the critical moment from the program logs as chosen by the investigator. At bottom right are the counterfactual conditions the investigator has specified.\vspace{3mm}}
    \label{fig:gui}
\end{figure*}

\begin{table*}[t]
\centering
\scriptsize
$$\begin{tabularx}{\textwidth}{c *{5}{Y}}
\toprule
     \multicolumn{2}{c|}{} & \multicolumn{3}{c|}{\vspace{1mm}\underline{timings (avg.~$n = 10$)}} & \multicolumn{1}{c}{}\\
    \textbf{model} & \textbf{output} & \textbf{symbolic ($s$)} & \textbf{solving ($s$)} & \textbf{total ($s$)} & \textbf{paths} \\[2pt]\midrule
    \fprop\rule{0pt}{2.5ex} & \multicolumn{5}{|l}{$\varphi_{fact}$, \textit{moved?}}\\\midrule
    standard & \cmark & 3.575 & 4.290e-03 & 4.162 & 1\\\midrule
    impatient & \cmark & 3.607 & 4.317e-03 & 4.193 & 1\\\midrule
    pathological & \cmark & 3.626 & 4.249e-03 & 4.212 & 1\\\midrule
    \cprop\rule{0pt}{2.5ex} & \multicolumn{5}{|l}{$\varphi^* \equiv \varphi_{fact}[(\texttt{agent1\_signal\_choice} = 2) \mapsto (\texttt{agent1\_signal\_choice} \in \{ 0, \, 1, \, 2 \})]$, \textit{always move?}}\\\midrule
    standard & \xmark & 3.979 & 2.371 & 7.754 & 3\\\midrule
    impatient & \cmark & 4.001 & 2.307 & 7.703 & 3\\\midrule
    pathological & \cmark & 3.958 & 2.326 & 7.681 & 3\\\midrule
    \cprop\rule{0pt}{2.5ex} & \multicolumn{5}{|l}{$\varphi^* \equiv \varphi_{fact}[(\texttt{agent1\_signal\_choice} = 2) \mapsto (\texttt{agent1\_signal\_choice} \in \{ 0, \, 1, \, 2 \})]$, \textit{ever not move?}}\\\midrule
    standard & \cmark & 4.015 & 2.428 & 7.849 & 3\\\midrule
    impatient & \xmark & 3.919 & 2.334 & 7.673 & 3\\\midrule
    pathological & \xmark & 3.966 & 2.352 & 7.718 & 3\\\midrule
    \cprop\rule{0pt}{2.5ex} & \multicolumn{5}{|l}{$\varphi^*[(\texttt{agent1\_pos\_x} = 1.376) \mapsto (1.0 \leq \texttt{agent1\_pos\_x} \leq 1.5) ]$, \textit{always move?}}\\\midrule
    standard & \xmark & 154.7 & 17.14 & 179.7 & 19\\\midrule
    impatient & \cmark & 207.6 & 4.622 & 220.1 & 19\\\midrule
    pathological & \xmark & 141.1 & 17.34 & 166.1 & 19\\\midrule
    \cprop\rule{0pt}{2.5ex} & \multicolumn{5}{|l}{$\varphi^*[(\texttt{agent1\_pos\_x} = 1.376) \mapsto (1.0 \leq \texttt{agent1\_pos\_x} \leq 1.5) ]$, \textit{ever not move?}}\\\midrule
    standard & \cmark & 133.0 & 54.40 & 195.2 & 19\\\midrule
    impatient & \xmark & 126.0 & 4.648 & 138.5 & 19\\\midrule
    pathological & \cmark & 254.2 & 17.47 & 279.5 & 19\\\midrule
    \cprop\rule{0pt}{2.5ex} & \multicolumn{5}{|l}{$\varphi^* \land (\texttt{agent2\_pos\_x} = 1.316) \land (\texttt{agent2\_pos\_z} = 0.378) \land \cdots$, \textit{always move?}}\\\midrule
    standard & \xmark & 8.995 & 4.111 & 16.74 & 3\\\midrule
    impatient & \cmark & 9.107 & 3.951 & 16.71 & 3\\\midrule
    pathological & \cmark & 9.037 & 3.913 & 16.54 & 3\\\midrule
    \cprop\rule{0pt}{2.5ex} & \multicolumn{5}{|l}{$\varphi^* \land (\texttt{agent2\_pos\_x} = 1.316) \land (\texttt{agent2\_pos\_z} = 0.378) \land \cdots$, \textit{ever not move?}}\\\midrule
    standard & \cmark & 8.483 & 4.029 & 16.33 & 3\\\midrule
    impatient & \xmark & 8.979 & 3.848 & 16.46 & 3\\\midrule
    pathological & \xmark & 9.087 & 3.941 & 16.70 & 3\\
\bottomrule
\end{tabularx}$$
\caption{Experimental results for our car crash case study.}   
\label{tbl:res}
\end{table*}

\section{The \soid{} Tool: Architecture and Case Studies}\label{sec:soid}

We implemented the counterfactual-guided logic exploration loop in our tool \soid{}, for \textbf{S}MT-based \textbf{O}racle for \textbf{I}nvestigating \textbf{D}ecisions.

The \soid{} tool is implemented in Python, and invokes the Z3 SMT solver~\cite{moura2008z3} for resolving queries. To begin, and outside of the scope of \soid{}, the investigator uses their knowledge of the harm under investigation to extract the factual trace $\tau^f$ from the logging infrastructure of $A$. 
\RB{Note that our tool assumes that both the $\tau^f$ and $A$ used in the analysis correspond to the real-world execution. Accountable logging~\cite{yoon2019adlp} and verifiable computation~\cite{parno2013pinocchio} can bolster confidence in these assumptions, and further that the program execution pathways being analyzed by \soid{} are those applicable in deployment and are not being manipulated by a `defeat device'~\cite{contag2017they}.} 
\RC{
At present \soid{} also assumes deterministic programs, though symbolic execution of randomized programs is an active area of formal methods research with developing tooling that could in the future be used to extend our method~\cite{susag2022symbolic}.
}

After extracting the trace the investigator specifies the (counter)factual query $\varphi(I)$ and behavior $\beta(\hat{D})$ using a Python library interface. Upon invocation, \soid{} symbolically executes $A$ to generate $\Pi(\hat{V})$. After the symbolic execution completes, $\soid{}$ formulates $\Phi$ as per Section~\ref{sec:rsandqs} and invokes Z3 to resolve the query. It then outputs the finding, as well as any model $\mathcal{M}$ in the event one exists due to a failed verification or successful counterfactual generation.

\paragraph{Research Questions.}\label{subsec:rqs}

We concretely evaluate \soid{} with respect to three quantitative research questions. They concern specific technical details of \soid{}'s efficiency:

\begin{itemize}
  \setlength\itemsep{0.5em}
    \item[RQ1)] Can \emph{ex post} analysis paired with symbolic execution to constrain to only feasible paths limit the consequences of the combinatorial explosion often present in verification tasks? 
\end{itemize}
More specifically, by logical duality, `will the car always move?' and `does there exist a scenario where the car does not move?' produce equivalent outputs. This is despite that the former is a universal `verification' query while the latter is an existential `counterfactual generation' query. So, we investigate whether `will the car always move?' is significantly more expensive to resolve than `does there exist a scenario where the car does not move?' for the same counterfactual family.
\begin{itemize}
  \setlength\itemsep{0.5em}
   \item[RQ2)] What is the relative additional cost of floating-point range constraints, given the importance of floating-point operations to statistical inference-based algorithmic decision making?
    \item[RQ3)] Do counterfactuals incorporating environmental conditions that are nonetheless irrelevant to agent decision making meaningfully increase the cost of a \soid{} execution, or does the symbolic execution successfully `ignore' them?
\end{itemize}

\subsection{Three Cars on the Stand: A Case Study}

In this section, we evaluate \soid{} on the crash example from Section~\ref{sec:motivating} (and Figure~\ref{fig:crash}). We pose and resolve the queries from the example:
\begin{enumerate}
 \item[at] $t_1^*$\\
    \vspace{-7mm}
    \begin{quote}
        $\counterfactual$ Could a different turn signal have led $\mathcal{A}$ to remain stationary? 
    \end{quote}
    \vspace{-2mm}
    \begin{quote}
        $\counterfactual$ If $\mathcal{A}$ had arrived before the other car, and that other car was not signaling a turn, would $\mathcal{A}$ have waited? (\emph{e.g.}, to `bait' the other car into passing in front of it?)  
    \end{quote}
\end{enumerate}
in a simulated driving environment, and show that \soid{} is able to produce \textsc{Facts} that distinguish between three different machine-learned self-driving car agents.

For our environment we employ Gym-Duckietown~\cite{gym_duckietown} with a simple intersection layout. A rendering of our example crash in our environment is given in Figure~\ref{fig:crash}. For our three agents, we used the same general C codebase, but used reinforcement learning -- specifically Q-learning~\cite{watkins1992q} -- to train three different versions of the decision model it invokes, each based on a different reward profile. Informally we deemed these reward profiles `standard', `impatient' and `pathological'. The `standard' profile is heavily penalized for crashing, but also rewarded for speed and not punished for moving without the right of way, so long as it is `safe'. The `impatient' profile is only rewarded for speed. The `pathological' profile is rewarded significantly for crashes, and minimally for speed to promote movement over nothing. The reward functions for these profiles are provided in Appendix~\ref{app:rewards}. The simulation environment is completely invisible to \soid{}, which only analyzes program executions on the basis of its code and logs. 

On top of Gym-Duckietown we designed and implemented a web GUI to enable non-expert interaction with \soid{}. GUIs that automatically generate representations of the driving environment are already deployed into semi-autonomous vehicles, such as those produced by Tesla. While simulating the environment, a drag-and-drop and button interface allows the user to manipulate the environment. by, \emph{e.g.}, introducing new cars, manipulating a car's position or angle, or changing a car's destination or which car possesses the right of way. After a factual trace plays out, a slider allows the investigator to select a step of the execution, before a drop-down and button interface allows specifying a counterfactual family and behavior (whether or not the car moved). We provide still images of the GUI's interfaces in Figure~\ref{fig:gui}.

\paragraph{Results.} The results of our benchmarks are summarized Table~\ref{tbl:res}. All statistics were gathered on an Intel Xeon CPU E5-2650 v3 @ 2.30GHz workstation with 64 GB of RAM. Each heading in Table~\ref{tbl:res} specifies a set of constraints $\varphi(\hat{I})$, and implicitly a behavior $\beta(\hat{D})$. The rows list the trained model invoked within $A$, the output of the evaluations, average timings, and the total number of feasible paths. Note that the \textbf{symbolic} and \textbf{solving} timings do not exactly sum to the \textbf{total} timing, due to some overhead. 

We find that \soid{} met each of our design goals on the way to successfully distinguishing the three cars, especially when enhanced by the GUI. The tool provides an interpretable and adaptive oracle allowing the investigator to query a sequence of counterfactuals without directly interacting with $A$ or the machine learned-model underlying it. Most of our queries resolved within $< 20s$, providing effective usability. The results of the queries demonstrate the distinctive behaviors expected of the three conflicting `purposes' as described in detail in \S\ref{sec:motivating}, allowing a thoughtful investigator to distinguish them as desired. With this data, we can also draw the following conclusions about our research questions.

\paragraph{RQ \#1.} We find that our universal and existential formulation took nearly identical time to resolve for all our evaluated families of counterfactuals. Moreover, we found that both positive and negative results for each query took near identical time as well. This finding highlights the benefit of \emph{ex post}, `human scale' analysis for combating combinatorial explosion. By limiting the scope of the queries to small, adaptively formulated relaxations, the set of program paths and the hardness of the resulting queries were kept manageable for the solver regardless of the nature of the query or its truth. This runs counter to the prevailing experience in \emph{ex ante} SMT-based program verification. 

\paragraph{RQ \#2.} We find the inclusion of a floating-point range query to notably increase the cost of solving, with a $\sim$6x increase in the number of feasible program paths and a $\sim$20-30x increase in the time required. As floating point ranges are a natural query constraint for machine learning-based systems, this increased in cost may be a significant limitation for practical deployments of \soid{}. Note that most of the increase in cost lies in symbolic execution ($\sim$50x increase) rather than solving ($\sim$10x increase). This encourages that advancements in floating point symbolic execution could greatly reduce this liability. The recent focus on SMT-based decision procedures for floating point/real-valued solving in the context of neural network verification~\cite{katz2017reluplex, katz2019marabou} could also be a source of relevant advances.

\paragraph{RQ \#3.} We find that \soid{} successfully `ignores' irrelevant information in the sense that no extraneous paths were explored, albeit at some increased cost in the solving of path constraint queries which were nonetheless enlarged by that information. Specifically, including a third car ($\texttt{agent2}$) outside of the intersection, but nonetheless present as an input to $A$, did not increase the number of paths found, but did increase the cost of solving by $\sim$2x. This is likely due to the increased size of the $\pi_i$ and $\Pi$ in the queried SMT formulas.

\captionsetup[figure]{format=hang}
\begin{figure*}[!t]
     \centering
     \begin{subfigure}[c]{0.48\textwidth}
          \centering
          \vspace{-50mm}   
          \includegraphics[width=0.9\textwidth]{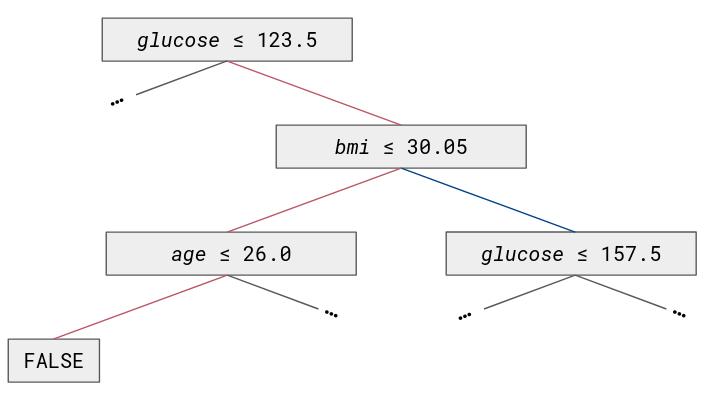}
     \end{subfigure}
     \begin{subfigure}[b]{0.48\textwidth}
          \centering
          \input{code/dt.tex}   
     \end{subfigure}
     \caption{Our decision tree example. At top, the relevant decision subtree for a misclassification based on health data, with the incorrect path taken in red -- and the correct branch missed in blue -- as the unit conversion bug leads to a significantly smaller BMI input than is correct. At bottom, the (otherwise correct) decision tree inference logic in C. \label{fig:dt}}
\end{figure*}

\subsection{Health Risk Decision Tree Misclassification}\label{app:dt}

\noindent To demonstrate that \soid{} is more general in application than to just cyberphysical systems, we also consider a second motivating example of incorrect statistical inference. We train a decision tree to infer the health risk status of individuals using the Pima Indians dataset, a classic example in counterfactuals due to~\cite{wachter2017counterfactual}. Notably, we consider a program $A$ with an implicit unit conversion bug: $A$ computes the BMI input to the decision tree using the height and weight parameters from its input. However, it is written to expect metric inputs in $kg$ and $m$, while the inputs are instead provided in the imperial $in$ and $lb$. This is a flaw of the software system in general. Both the decision tree and program themselves are correct, but end-to-end the system misclassifies many inputs, as for the same quantities $(kg/m^2) \gg (lb/in^2)$. 

Unlike statistical counterfactual methods like those of~\cite{wachter2017counterfactual,mothilal2020explaining} which only analyze the (correct) decision model, the end-to-end nature of \soid{} allows it to analyze everything, including the conversion bug. Figure~\ref{fig:dt} displays the inference code and incorrect decision due to the conversion error.

We then ran a small case study on this decision tree health risk misclassification example. The results of our benchmarks are summarized Table~\ref{tbl:sires}, and were gathered on the same Intel Xeon CPU E5-2650 v3 @ 2.30GHz workstation with 64 GB of RAM. In additional to a simple factual verification query as a baseline, we posed a single counterfactual query:
\begin{enumerate}
 \item[at] $t^*$\\
    \vspace{-7mm}
    \begin{quote}
        $\counterfactual$ Does there exist a weight for which the instance is classified as high risk? 
    \end{quote}
\end{enumerate}
The results show \soid{} is able to efficiently resolve the counterfactual in the positive. Of our empirical research questions only RQ \#2 is directly applicable.

\paragraph{RQ \#2.} We find the additional cost of the floating-point operations to still enable quick solving, even though those operations lead to SMT constraints incorporating exponentiation and division operations over floating-point operations, in addition to the (bounded-depth) recursion over floating-point comparisons needed for the decision tree inference itself.

\begin{table*}[t]
    \centering
    \scriptsize
    $$\begin{tabularx}{\textwidth}{c *{5}{Y}}
\toprule
  \multicolumn{2}{c|}{} & \multicolumn{3}{c|}{\vspace{1mm}\underline{timings (avg.~$n = 10$)}} & \multicolumn{1}{c}{}\\
    \textbf{model} & \textbf{output} & \textbf{symbolic ($s$)} & \textbf{solving ($s$)} & \textbf{total ($s$)} & \textbf{paths} \\[2pt]\midrule
    \fprop\rule{0pt}{2.5ex} & \multicolumn{5}{|l}{$\varphi_{fact}$, \textit{low risk?}}\\\midrule
    dt & \cmark & 0.746 & 4.896e-03 & 0.812 & 1\\\midrule
    \cprop\rule{0pt}{2.5ex} & \multicolumn{5}{|l}{$\varphi^* \equiv \varphi_{fact}[(\texttt{weight} = 249.973) \mapsto \top]$, \textit{ever high risk?}}\\\midrule
    dt & \cmark & 2.277 & 1.655 & 4.009 & 2\\
\bottomrule
\end{tabularx}$$
\caption{Experimental results for our decision tree misclassification case study.}   
\label{tbl:sires}
\end{table*}

\section{Conclusion}\label{sec:conclusion}

We conclude by highlighting some promising future directions for extending \acronym{}{} and \soid{} to both support more complex agents and make the investigatory process more intuitive.

\paragraph{Exploiting Counterfactual-Guided Abstraction Refinement.} As described, a promising direction is to add additional automation into \soid{} by using the abstraction refinement literature to generate informative representations of critical states and logic that recur across many queries. 

\paragraph{Supporting DNNs.} Many modern machine-learned agents rely on models built out of deep neural network (DNN) architectures. Extending \soid{} to support such agents -- most likely by relying on SMT-based neural network verifiers as subroutines~\cite{katz2017reluplex, katz2019marabou} -- is an important open direction for increasing the utility of our method and tools.

\paragraph{Alternative Generation of $\Pi(\hat{V})$.} Symbolic execution is no the only way to generate $\Pi(\hat{V})$. In particular, predicate transformer semantics~\cite{jhala2009software} provide another direction for generating decision logic representations compatible with SMT. Comparing the relative strengths of various methods may show that other approaches are superior, \emph{e.g.}, weakest precondition generation could be used to precompute a single large $\Pi(\hat{V})$ for all possible program executions, amortizing the cost of finding $\Pi(\hat{V})$ in exchange for an increase in the effort required to solve $\Phi$.

\paragraph{Programming Counterfactuals.} Although \soid{} is adaptive, that does not necessarily mean it needs to be interactive. A further possible direction would be to design a counterfactual calculus as the basis for a programming language that would invoke \soid{} to solve for elements of the \textsc{Facts} as part of the semantics. Such a language could potentially be the basis for formalizing certain legal regimes for which counterfactual analysis forms a critical component. A related direction would be to integrate with a scenario specification language like SCENIC~\cite{fremont2019scenic}, and the VerifAI project more generally~\cite{dreossi2019verifai}, to add another layer of capability onto the specification of families of counterfactual scenarios.

\section*{Acknowledgements} 

The authors thank Gideon Yaffe for many helpful conversations, Man-Ki Yoon for his assistance in implementing an earlier simulated driving environment, and Cristian Cadar and Daniel Liew for their guidance on successfully using Klee-Float for symbolic execution of our experiments. This work was supported by the Office of Naval Research (ONR) of the United States Department of Defense through a National Defense Science and Engineering Graduate (NDSEG) Fellowship, by the State Government of Styria, Austria – Department Zukunftsfonds Steiermark, by EPSRC grant no EP/R014604/1, and by NSF awards CCF-2131476, CCF-2106845, and CCF-2318974. The authors would also like to thank the Isaac Newton Institute for Mathematical Sciences, Cambridge, for support and hospitality during the programme Verified Software where work on this paper was undertaken.

%% Bibliography
\bibliographystyle{plain}
\bibliography{bibfile}

%% Appendix
\appendix

\setlength{\belowcaptionskip}{-10pt}

\section{Additional Motivation}\label{app:motivate}

We consider a second motivating example of an \emph{ex post} investigation and analysis of a simulated self-driving car crash, as diagrammed in Figure~\ref{fig:me}. This example more thoroughly demonstrates how counterfactual analysis can clarify how an agent's decision making reacts to changes in the causal chain of events. For both cognitive and computational agents, traditional models of intention tie planning to commitment: if an agent intends an outcome, then it will only reconsider their behavior when that outcome becomes infeasible or counterproductive~\cite{bratman1987intention}. Counterfactuals then  allow querying and interpreting under what hypotheticals an agent can or will alter their behavior.

\captionsetup[figure]{format=hang}
\begin{figure}[!t]
     \centering
     \begin{subfigure}[t]{0.34\textwidth}
          \centering
          \resizebox{1\linewidth}{!}{\definecolor{qqttcc}{rgb}{0.59607, 0.5529, 0.7608}
\definecolor{ffffzz}{rgb}{0.902, 0.827, 0}
\definecolor{zzttqq}{rgb}{0.9451,0.6353,0.2510}
\definecolor{ffffff}{rgb}{1,1,1}
\definecolor{qqwuqq}{rgb}{0.949, 0.949, 0.949}
\begin{tikzpicture}[line cap=round,line join=round,>=triangle 45,x=1cm,y=1cm]
\clip(-4.584983792540159,-4.584983792540159) rectangle (4.584983792540159, 4.584983792540159);
\fill[line width=2pt,color=qqwuqq,fill=qqwuqq,fill opacity=0.4] (-12,12) -- (-12,-12) -- (12,-12) -- (12,12) -- cycle;
\fill[line width=2pt,fill=black,fill opacity=0.1] (-12,2.25) -- (-12,-2.25) -- (12,-2.25) -- (12,2.25) -- cycle;
\fill[line width=2pt,fill=black,fill opacity=0.1] (-2.25,12) -- (-2.25,-12) -- (2.25,-12) -- (2.25,12) -- cycle;
\draw [line width=2pt,color=qqwuqq] (-12,12)-- (-12,-12);
\draw [line width=2pt,color=qqwuqq] (-12,-12)-- (12,-12);
\draw [line width=2pt,color=qqwuqq] (12,-12)-- (12,12);
\draw [line width=2pt,color=qqwuqq] (12,12)-- (-12,12);
\draw [line width=2.4pt,dash pattern=on 1pt off 1pt,color=ffffff] (0,12)-- (0,-12);
\draw [line width=2.4pt,dash pattern=on 1pt off 1pt,color=ffffff] (-12,0)-- (12,0);
\draw [line width=1.2pt,color=ffffff] (-2.25,3)-- (2.25,3);
\draw [line width=1.2pt,color=ffffff] (2.25,2.25)-- (-2.25,2.25);
\draw [line width=1.2pt,color=ffffff] (-2.25,-2.25)-- (2.25,-2.25);
\draw [line width=1.2pt,color=ffffff] (2.25,-3)-- (-2.25,-3);
\draw [line width=1.2pt,color=ffffff] (-2.25,3)-- (2.25,2.25);
\draw [line width=1.2pt,color=ffffff] (2.25,3)-- (-2.25,2.25);
\draw [line width=1.2pt,color=ffffff] (-2.25,-2.25)-- (2.25,-3);
\draw [line width=1.2pt,color=ffffff] (-2.25,-3)-- (2.25,-2.25);
\draw [line width=1.2pt,color=ffffff] (-2.25,2.25)-- (-3,-2.25);
\draw [line width=1.2pt,color=ffffff] (-2.25,-2.25)-- (-3,2.25);
\draw [line width=1.2pt,color=ffffff] (-2.25,2.25)-- (-2.25,-2.25);
\draw [line width=1.2pt,color=ffffff] (2.25,-2.25)-- (2.25,2.25);
\draw [line width=1.2pt,color=ffffff] (2.25,2.25)-- (3,-2.25);
\draw [line width=1.2pt,color=ffffff] (3,-2.25)-- (3,2.25);
\draw [line width=1.2pt,color=ffffff] (3,2.25)-- (2.25,-2.25);
\draw [shift={(-2.37025,-1.13)},line width=1.2pt,color=zzttqq,fill=zzttqq,pattern=north east lines,pattern color=zzttqq]  (0,0) --  plot[domain=-1.5445036858989027:1.5445036858989023,variable=\t]({1*0.7702662283262843*cos(\t r)+0*0.7702662283262843*sin(\t r)},{0*0.7702662283262843*cos(\t r)+1*0.7702662283262843*sin(\t r)}) -- cycle ;
\draw [line width=2pt,dash pattern=on 1pt off 1pt,color=zzttqq] (-12,-1.13)-- (-2.35,-1.13);
\draw [line width=1.2pt,color=ffffff] (-3,2.25)-- (-3,-2.25);
\draw [shift={(1.13,-4.020266666666667)},line width=1.2pt,color=qqttcc,fill=qqttcc,pattern=north east lines,pattern color=qqttcc]  (0,0) --  plot[domain=0.026314270945512044:3.115278382644281,variable=\t]({1*0.7702666666666665*cos(\t r)+0*0.7702666666666665*sin(\t r)},{0*0.7702666666666665*cos(\t r)+1*0.7702666666666665*sin(\t r)}) -- cycle ;
\draw [line width=2pt,dash pattern=on 1pt off 1pt,color=qqttcc] (1.13,-12)-- (1.13,-4);
\draw [shift={(-0.020266666666666655,-1.13)},line width=1.2pt,color=zzttqq,fill=zzttqq,fill opacity=0.56]  (0,0) --  plot[domain=-1.544482055849385:1.5444820558493852,variable=\t]({1*0.7702666666666668*cos(\t r)+0*0.7702666666666668*sin(\t r)},{0*0.7702666666666668*cos(\t r)+1*0.7702666666666668*sin(\t r)}) -- cycle ;
\draw [line width=2pt,dash pattern=on 1pt off 1pt,color=zzttqq] (-1.6,-1.125)-- (0,-1.13);
\draw [shift={(1.13,-2.1702666666666666)},line width=1.2pt,color=qqttcc,fill=qqttcc,fill opacity=0.5]  (0,0) --  plot[domain=0.026314270945511475:3.1152783826442816,variable=\t]({1*0.7702666666666664*cos(\t r)+0*0.7702666666666664*sin(\t r)},{0*0.7702666666666664*cos(\t r)+1*0.7702666666666664*sin(\t r)}) -- cycle ;
\draw [line width=2.4pt,dash pattern=on 1pt off 1pt,color=qqttcc] (1.13,-2.15)-- (1.13,-3.25);
\begin{scriptsize}
\draw [fill=ffffzz] (-1.6,-1.515) ++(-4.5pt,0 pt) -- ++(4.5pt,4.5pt)--++(4.5pt,-4.5pt)--++(-4.5pt,-4.5pt)--++(-4.5pt,4.5pt);
\end{scriptsize}
\end{tikzpicture}}
     \end{subfigure}
     \begin{subfigure}[t]{0.35\textwidth}
          \centering
          \resizebox{1\linewidth}{!}{\definecolor{ffffzz}{rgb}{0.902, 0.827, 0}
\definecolor{zzttqq}{rgb}{0.9451,0.6353,0.2510}
\definecolor{qqttcc}{rgb}{0.59607, 0.5529, 0.7608}
\definecolor{qqwwtt}{rgb}{0.553, 0.761, 0.596}
\definecolor{yqqqyq}{rgb}{0.761, 0.596, 0.553}
\definecolor{ffffff}{rgb}{1,1,1}
\definecolor{qqwuqq}{rgb}{0.949, 0.949, 0.949}
\begin{tikzpicture}[line cap=round,line join=round,>=triangle 45,x=1cm,y=1cm]
\clip(-4.584983792540159,-4.584983792540159) rectangle (4.584983792540159, 4.584983792540159);
\fill[line width=2pt,color=qqwuqq,fill=qqwuqq,fill opacity=0.4] (-12,12) -- (-12,-12) -- (12,-12) -- (12,12) -- cycle;
\fill[line width=2pt,fill=black,fill opacity=0.1] (-12,2.25) -- (-12,-2.25) -- (12,-2.25) -- (12,2.25) -- cycle;
\fill[line width=2pt,fill=black,fill opacity=0.1] (-2.25,12) -- (-2.25,-12) -- (2.25,-12) -- (2.25,12) -- cycle;
\fill[line width=2pt,color=yqqqyq,fill=yqqqyq,fill opacity=0.11] (3,3) -- (-3,3) -- (3,-3) -- cycle;
\draw [line width=2pt,color=qqwuqq] (-12,12)-- (-12,-12);
\draw [line width=2pt,color=qqwuqq] (-12,-12)-- (12,-12);
\draw [line width=2pt,color=qqwuqq] (12,-12)-- (12,12);
\draw [line width=2pt,color=qqwuqq] (12,12)-- (-12,12);
\draw [line width=2.4pt,dash pattern=on 1pt off 1pt,color=ffffff] (0,12)-- (0,-12);
\draw [line width=2.4pt,dash pattern=on 1pt off 1pt,color=ffffff] (-12,0)-- (12,0);
\draw [line width=1.2pt,color=ffffff] (-2.25,3)-- (2.25,3);
\draw [line width=1.2pt,color=ffffff] (2.25,2.25)-- (-2.25,2.25);
\draw [line width=1.2pt,color=ffffff] (-2.25,-2.25)-- (2.25,-2.25);
\draw [line width=1.2pt,color=ffffff] (2.25,-3)-- (-2.25,-3);
\draw [line width=1.2pt,color=ffffff] (-2.25,3)-- (2.25,2.25);
\draw [line width=1.2pt,color=ffffff] (2.25,3)-- (-2.25,2.25);
\draw [line width=1.2pt,color=ffffff] (-2.25,-2.25)-- (2.25,-3);
\draw [line width=1.2pt,color=ffffff] (-2.25,-3)-- (2.25,-2.25);
\draw [line width=1.2pt,color=ffffff] (-3,2.25)-- (-3,-2.25);
\draw [line width=1.2pt,color=ffffff] (-2.25,2.25)-- (-3,-2.25);
\draw [line width=1.2pt,color=ffffff] (-2.25,-2.25)-- (-3,2.25);
\draw [line width=1.2pt,color=ffffff] (-2.25,2.25)-- (-2.25,-2.25);
\draw [line width=1.2pt,color=ffffff] (2.25,-2.25)-- (2.25,2.25);
\draw [line width=1.2pt,color=ffffff] (2.25,2.25)-- (3,-2.25);
\draw [line width=1.2pt,color=ffffff] (3,-2.25)-- (3,2.25);
\draw [line width=1.2pt,color=ffffff] (3,2.25)-- (2.25,-2.25);
\draw [shift={(1.13,-1.2702666666666669)},line width=1.2pt,color=qqttcc,fill=qqttcc,pattern=north east lines,pattern color=qqttcc]  (0,0) --  plot[domain=0.026314270945511763:3.1152783826442816,variable=\t]({1*0.7702666666666668*cos(\t r)+0*0.7702666666666668*sin(\t r)},{0*0.7702666666666668*cos(\t r)+1*0.7702666666666668*sin(\t r)}) -- cycle ;
\draw [line width=2pt,dash pattern=on 1pt off 1pt,color=qqttcc] (1.13,-1.25)-- (1.13,-12);
\draw [shift={(-0.8,1.7702666666666667)},line width=2pt,color=zzttqq,fill=zzttqq,fill opacity=0.57]  (0,0) --  plot[domain=3.1679069245353046:6.256871036234075,variable=\t]({1*0.7702666666666664*cos(\t r)+0*0.7702666666666664*sin(\t r)},{0*0.7702666666666664*cos(\t r)+1*0.7702666666666664*sin(\t r)}) -- cycle ;
\draw [shift={(1.13,1.8297333333333339)},line width=1.2pt,color=qqttcc,fill=qqttcc,fill opacity=0.57]  (0,0) --  plot[domain=0.026314270945510896:3.1152783826442825,variable=\t]({1*0.7702666666666668*cos(\t r)+0*0.7702666666666668*sin(\t r)},{0*0.7702666666666668*cos(\t r)+1*0.7702666666666668*sin(\t r)}) -- cycle ;
\draw [line width=2.4pt,dotted,color=qqwwtt] (-0.5,2.425)-- (0.52,2.43);
\draw [line width=2pt,dash pattern=on 1pt off 1pt,color=qqttcc] (1.13,1.85)-- (1.13,-0.5);
\begin{scriptsize}
\draw [color=yqqqyq] (3,3)-- ++(-4.5pt,0 pt) -- ++(9pt,0 pt) ++(-4.5pt,-4.5pt) -- ++(0 pt,9pt);
\draw [fill=qqwwtt,shift={(-0.5,2.425)},rotate=270] (0,0) ++(0 pt,6.75pt) -- ++(5.84567147554496pt,-10.125pt)--++(-11.69134295108992pt,0 pt) -- ++(5.84567147554496pt,10.125pt);
\draw [fill=ffffzz] (-0.415,1) ++(-4.5pt,0 pt) -- ++(4.5pt,4.5pt)--++(4.5pt,-4.5pt)--++(-4.5pt,-4.5pt)--++(-4.5pt,4.5pt);
\draw [fill=qqwwtt,shift={(0.52,2.43)},rotate=270] (0,0) ++(0 pt,6.75pt) -- ++(5.84567147554496pt,-10.125pt)--++(-11.69134295108992pt,0 pt) -- ++(5.84567147554496pt,10.125pt);
\end{scriptsize}
\end{tikzpicture}}  
     \end{subfigure}
     \caption{A pair of self-driving car crash scenarios.\newline(L): \emph{Ex. 1}. A broadside crash between an agent car $\mathcal{A}$ (blue semicircle) and another car (orange semicircle). See Figure~\ref{fig:crash} for a rendering of this scenario in our GUI.\newline(R): \emph{Ex. 2}. An agent car $\mathcal{A}$ (blue semicircle) hits a pedestrian (green triangle) in an environment with an obscuring truck (orange semicircle) and an automated intersection (red cross). \label{fig:me}}
\end{figure}

Suppose a preliminary investigation establishes the following facts regarding a crash: as $\mathcal{A}$ (blue semicircle) navigated the intersection, a pedestrian (green triangle) walked against the signal. At time $t^*_1$ (striped $\mathcal{A}$) the pedestrian was obscured from $\mathcal{A}$ by the turning truck (orange semicircle). At that time an intersection manager (or IM, red cross) misclassified the pedestrian as a cyclist, and accordingly sent a misleading warning to $\mathcal{A}$. At some $t^*_2$, $\mathcal{A}$'s own sensors acquired and correctly identified the pedestrian, but the car was unable avoid the crash at time $t^*_3$ (solid $\mathcal{A}$).\footnote{Related scenarios appear in marketing materials for self-driving car technology~\cite{continental}.} 

Consider the following questions:
\begin{enumerate}
    \item[at] $t_1^*$\\
    \vspace{-8mm}
    \begin{quote}
        $\to$ Did the receipt of the IM's message lead to a decision to brake?
    \end{quote}
    \vspace{-2mm}
    \begin{quote}
        $\counterfactual$ Would a pedestrian warning have lead to a decision to (more sharply) brake?
    \end{quote}
    \item[at] $t_2^*$\\
    \vspace{-8mm}
    \begin{quote}
        $\to$ Did the observation of the pedestrian lead to a decision to (more sharply) brake?
    \end{quote}
    \vspace{-2mm}
    \begin{quote}
        $\counterfactual$ If $\mathcal{A}$'s sensors also misidentified the pedestrian as a cyclist, would that have led to a decision to (less sharply) brake?
    \end{quote}
\end{enumerate}
Note again their adaptive nature: the exact formulation of the later questions depend on the answer to the first. These queries explore under what conditions the car will reconsider transiting the intersection, and the relative responsibilities of the three agents. The answers (\emph{yes}, \emph{yes}, \emph{yes}, \emph{no}) support the IM being significantly culpable through its misidentification, as the agent's decisions are consistent with believing it braked sufficiently to avoid a (faster moving) cyclist. In contrast, (\emph{no}, \emph{no}, \emph{yes}, \emph{no}) support $\mathcal{A}$ as uniquely culpable among the ADMs, as the misidentification by the IM had limited influence on its decisions.

\section{Reward Functions}\label{app:rewards}

We document the reward profiles we used to train the `standard', `impatient', and `pathological' cars respectively, up to some minor edits for clarity. Note that in order to cleanly test our intersection dynamics every car attempts to avoid tailgating, in order to reduce spurious collisions on entry.

\begin{figure}[h]
    \centering
    \input{code/standard.tex}
    \caption{The `standard' reward profile.}
    \label{fig:standard}
\end{figure}
\begin{figure}[h]
    \centering
    \input{code/impatient.tex}
    \caption{The `impatient' reward profile.}
    \label{fig:impatient}
\end{figure}
\begin{figure}[h]
    \centering
    \input{code/pathological.tex}
    \caption{The `pathological' reward profile.}
    \label{fig:pathological}
\end{figure}

\end{document}